\documentclass[preprint,review, 12pt,authoryear]{elsarticle}

\usepackage{amssymb}
\usepackage{amsmath}
\usepackage{amsthm}

\usepackage{epstopdf}
\usepackage{subfigure}

\DeclareMathOperator*{\argmax}{arg\,max}

\usepackage{booktabs}
\usepackage{caption}

\usepackage{multirow}
\usepackage{makecell}

\usepackage{color}
\usepackage{hyperref}

\theoremstyle{plain}
\newtheorem{theorem}{Theorem}[section]
\newtheorem{lemma}[theorem]{Lemma}

\theoremstyle{definition}

\theoremstyle{remark}

\journal{Journal of Air Transport Management}

\usepackage{fancyhdr}
\fancypagestyle{pprintTitle}{%
  \fancyhf{}  
  \fancyfoot[R]{\footnotesize \today}  
}

\begin{document}

\begin{frontmatter}



\title{Runway capacity expansion planning for public airports under demand uncertainty} 


\author[label1]{Ziyue Li} 
\author[label2]{Joseph Y.J. Chow}
\author[label1]{Qianwen Guo}

\affiliation[label1]{organization={Department of Civil and Environmental Engineering, FAMU-FSU College of Engineering,  Florida State University},
            addressline={2525 Pottsdamer Street}, 
          city={Tallahassee},
           postcode={32310}, 
           state={FL},
            country={USA}}
\affiliation[label2]{organization={C2SMARTER University Transportation Center, Department of Civil and Urban Engineering, New York University Tandon School of Engineering},
           addressline={6 Metrotech Center}, 
            city={Brooklyn},
           postcode={11021}, 
            state={NY},
           country={USA}}

\begin{abstract}
Flight delay is a significant issue affecting air travel. The runway system, frequently falling short of demand, serves as a bottleneck. As demand increases, runway capacity expansion becomes imperative to mitigate congestion. However, the decision to expand runway capacity is challenging due to inherent uncertainties in demand forecasts. This paper presents a novel approach to modeling air traffic demand growth as a jump diffusion process, incorporating two layers of uncertainty: Geometric Brownian Motion (GBM) for continuous variability and a Poisson process to capture the impact of crisis events, such as natural disasters or public health emergencies, on decision-making. We propose a real options model to jointly evaluate the interrelated factors of optimal runway capacity and investment timing under uncertainty, with investment timing linked to trigger demand.  The findings suggest that increased uncertainty indicates more conservative
decision-making. Furthermore, the relationship between optimal investment timing and expansion size is complex: if the expansion size remains unchanged, the trigger demand decreases as the demand growth rate increases; if the expansion size experiences a jump, the trigger demand also exhibits a sharp rise. This work provides valuable insights for airport authorities for informed capacity expansion decision-making.
\end{abstract}



\begin{keyword}



Airport runway \sep Capacity expansion \sep Investment timing \sep Expansion size \sep Real options

\end{keyword}

\end{frontmatter}



\section{Introduction}

Air traffic demand has grown over the past few decades, resulting in a demand-capacity imbalance that frequently causes flight delays or even cancellations (\cite{sheng2019modeling, li2022modeling}). In particular, flight delay is a significant issue in U.S. air travel, with an estimated total economic cost of 23.7 billion dollars in 2016, rising to 33 billion dollars in 2019 (\cite{lukacs2020cost}). The issue has prompted extensive research in airport demand and capacity management (see, for example, \cite{madas2008airport, farhadi2014runway, gillen2016airport, di2020critical, polater2020airports, de2020airport, scala2021optimization}). \cite{jacquillat2018roadmap} summarize the approaches in airport demand and capacity management, classifying them into three main categories: infrastructure expansion, operational enhancements, and demand management.

The challenge of accommodating demand growth arises across various areas within an airport, including delays for passengers at security checkpoints (\cite{bullock2010automated}), limitations in the baggage handling system (\cite{kim2017balancing}), and queuing of aircraft on runways (\cite{rodriguez2019assessment}). The runway system, however, is the most critical bottleneck, as its capacity frequently falls short of demand. As the weakest link in the process chain of a flight, runway capacity limit becomes especially important (\cite{itr2007runway}).

Although the runway throughput rate may vary depending on several operational factors, such as airport physical layout (\cite{sumathi2018airport}), weather conditions (\cite{rodriguez2022impact}), and arrival and departure strategies (\cite{gilbo2013enhanced}), it is clear that the number of runways plays a vital role in determining the maximum potential capacity of the airport system. All else being equal, the more runways available, the more movements can be operated simultaneously, thereby enhancing overall capacity (\cite{jacquillat2018roadmap}).

Initially, some airports operated with a single runway for various reasons, including limited funding, small initial demand, and restricted space. As air traffic demand grows over time, these airports may invest in additional runways to accommodate more flights and passengers. Singapore Changi Airport (SIN) is an example of this evolution, having expanded from a single runway to meet increasing demand. SIN airport opened in 1981 with one runway, and a new parallel runway was completed in 2004 (\cite{SIN2024runway}). The expansion enabled SIN airport to handle more flights and larger aircraft, making it one of the busiest airports in the world today.

In the literature, there are primarily two approaches to addressing congestion within airport systems. The first approach is demand management, which improves mechanisms for airport resource pricing and allocation to achieve higher efficiency and welfare within the existing capacity limitations (\cite{zhang2010airport, xiao2013demand, sheng2015slot, wan2015airport}). In this context, airports, airlines, and passengers form a market. Airports charge congestion tolls to airlines as a means of demand management and a source of investment financing. Airlines then set higher fares, which pass the cost of congestion onto passengers. Higher fares also curtail passenger demand, thereby reducing congestion. The second approach is infrastructure expansion, which involves investing in airport capacity to achieve a better match between supply and demand. Two methodologies are commonly applied within this approach. The first utilizes stochastic programming (\cite{solak2009airport, sun2015stochastic, sun2016capacity}) where stochastic programming models are converted into mixed integer nonlinear programs, for which efficient algorithms are proposed. The second methodology employs real options models (\cite{xiao2017modeling, zheng2020airline, balliauw2020expanding}). Drawing from the financial concept of call options, real options provide the airport authority with the right to invest in additional capacity when demand reaches a specified level within a particular time frame, granting the flexibility to take this action without the obligation to do so. This methodology is also widely used in other transportation investment planning areas, including urban transit (\cite{GUO2018288, li2015transit}), railway systems (\cite{couto2015high, guo2018time, sang2019design, guo2023investment}) and highway networks (\cite{chow2011network, chow2011real, chow2011network2}).

Regardless of the methodology applied, a delay function is essential for quantifying the costs of congestion. However, most studies assume that demand cannot exceed capacity (\cite{zhang2006airport, sun2015stochastic, balliauw2020port, balliauw2020expanding}). \cite{zhang2006airport} used a specific delay function based on steady-state queuing theory, where the delay cost goes to infinity when demand approaches capacity. \cite{balliauw2020expanding} applied a different delay function but also asserted that ``output is bound by the designed capacity of the airport system", setting output equal to capacity when demand is larger than capacity to prevent the delay cost from growing unrealistically large. Nonetheless, when demand exceeds capacity, the airport incurs additional, but not infinite, costs associated with schedule coordination or slot control (\cite{jacquillat2018roadmap}). In practice, demand can and does exceed capacity; for instance, San Diego International Airport (SAN) has reached its full capacity of takeoffs and landings, years earlier than originally forecast (\cite{becker2019SAN}). This underscores the need for a well-formulated delay function, particularly for scenarios where demand approaches or exceeds capacity.

Moreover, uncertainties inherent in air demand growth affect the decision-making process (\cite{sun2015stochastic}). Nevertheless, varying levels of uncertainty often receive insufficient attention. Beyond random fluctuations, the effects of crisis events, such as natural disasters or public incidents like COVID-19, should also be accounted for, as they can lead to abrupt declines in demand growth (\cite{tirtha2022airport, guo2023investment}), leading to more conservative decision-making. An analytical model quantifying these uncertainties is necessary to assist airport authorities in investment decisions regarding capacity expansion, given that such investment is typically costly and irreversible. If airport authorities invest in runway capacity too early or with excessive expansion, they may incur unnecessary costs that could have been avoided by postponing the investment or through a smaller expansion size. Conversely, investing too late or with inadequate expansion may lead to congestion and increased costs due to rising demand for airport services. Therefore, two interconnected decisions, when to invest and how to expand, should be considered simultaneously.

Considering all of the above, this paper proposes a real options model with a jump process to address the airport runway capacity expansion problem. Specifically, we determine the optimal investment timing and expansion size for capacity expansion, aiming to minimize the costs associated with the airport runway system. The contributions of this work include the following:

\begin{enumerate}
    \item Explicit solutions for the timing and size decisions regarding capacity expansion are obtained. Compared with \cite{balliauw2020expanding}, this study provides analytical solutions for both deterministic and stochastic models.
    \item Disruptions affecting air traffic demand are taken into account, with their occurrence probability and impact magnitude on demand explicitly modeled. A jump diffusion model is introduced to capture such random events. 
    \item A modified delay function is proposed to address scenarios where airport demand approaches or even exceeds capacity, avoiding the delay function from becoming unrealistically large.
    \item Sensitivity analyses are conducted to examine how investment timing and expansion size vary with the airport's initial runway capacity, demand growth rate, and levels of uncertainty. The results provide more detailed insights compared to those presented in \cite{balliauw2020expanding}.
\end{enumerate}

The paper is structured as follows. Section \ref{sec: Problem Description and Model Formulation} outlines the problem and the model formulation, including the model setup, cost function, deterministic dynamic model, and stochastic dynamic model. Section \ref{sec: Model Applications} presents the results from numerical experiments, comparing different airports with varying characteristics. Finally, conclusions are drawn in Section \ref{sec: Conclusions and Future Studies}.

\section{Problem description and model formulation}
\label{sec: Problem Description and Model Formulation}
\subsection{Model setup}

Throughout this paper, we examine the cost structure of the public airport runway system with the goal of minimizing system costs by evaluating the decision to invest in runway capacity expansion. Additionally, we determine the optimal investment timing (when to invest) and expansion size (the increment of capacity expansion) for such investment. A description of the considered cost components is provided in Section \ref{sec:cost_function}. In Section \ref{sec:deterministic_dynamic_model}, a deterministic dynamic model is derived, while a stochastic dynamic model is proposed in Section \ref{sec:stochastic_dynamic_model}.

The notations and units for the parameters used in the following analyses are summarized in Table \ref{tab:variable_definition}. To distinguish between the two models for convenience, let $Q_t$ denote the demand at time $t$ in the deterministic model, and $Q(t)$ represent the demand in the stochastic model, where $Q(t)$ follows a stochastic process. The symbol $q$ refers to a fixed demand level without randomness. In this paper, demand is assumed to be an exogenous process. Meanwhile, $T$ denotes the investment timing in the deterministic model, while $\tau$ represents the investment timing in the stochastic model, where $\tau$ is a random variable.

\begin{table}[ht]
\caption{Variable Definition.}
\centering 
~\\
\resizebox{0.98\textwidth}{!}{%
$
{\begin{tabular}{llll}
\toprule
Notation & Definition & \makecell[l]{Baseline \\ Value} & Unit  \\\midrule
$Q_t, Q(t), q$ & Air traffic demand & -- & operation/hour\\     
$c$ & Unit operating cost & 1,000 & \$/operation\\
$c_h$ & Unit capacity holding (inventory) cost & 500 & \$/operation\\ 
$A$ & Monetary scaling factor of delay cost & 50,000 & \$/operation\\  
$\alpha$ & Specific constant of the delay function & 2 & --\\  
$\beta$ & Specific constant of the delay function & 3 & --\\  
$\mu$ & Specific constant of the delay function & 1 & --\\  
$\nu$ & Capacity utilization & -- & --\\  
$f$ & Fixed capital cost of capacity expansion & 8e6 & \$\\  
$v$ & Variable capital cost of capacity expansion & 20 & (\$$\times$hour)/operation\\  
$\eta$ & Air traffic demand growth rate & 0.02 & --\\  
$N_p$ & Operation hours per year & 4,000 & hour/year\\ 
$\rho$ & Annual discount rate & 0.07 & --\\ 
$T, \tau$ & Investment timing & -- & year\\ 
$K_0$ & Initial runway capacity & 40 & operation/hour\\ 
$\Delta K$ & Expansion size (increment of runway capacity expansion) & -- & operation/hour\\ 
$\sigma$ & Volatility rate (for the stochastic model) &  0.05 & --\\ 
$\lambda$ & Arrival rate of the Poisson process (for the stochastic model) & 0.05 & --\\ 
\bottomrule
\end{tabular}}
$%
}
\label{tab:variable_definition}
\end{table}

\subsection{Cost function}
\label{sec:cost_function}
We consider the generalized costs associated with the airport runway system, encompassing costs incurred by both the airport and the aircraft utilizing the runway system. Aircraft are charged for runway usage, generating revenue for the airport. However, this revenue is treated as an internal system variable and is excluded from the generalized costs analyzed in this study, as the focus is on public airports (see, e.g., \cite{zhang2010airport}). This approach differs from that of private airports, whose objective is profit maximization, as discussed in \cite{balliauw2020expanding}. While the specific cost components and their values associated with an airport's runway system can vary by airport, they generally include operating costs, capacity-related (inventory) holding costs, and delay-related costs. Additionally, there is a one-time investment cost for the capacity expansion project. 

First, we calculate the hourly system costs, which are comprised of operating costs, capacity holding costs, and delay costs. The hourly system cost function is expressed as follows:
\begin{equation}
\label{eq:system_cost_def}
    C(q) = cq + c_h K + D(q, K).
\end{equation}
The first term on the right-hand side of Eq. (\ref{eq:system_cost_def}) represents operating costs --- expenses incurred to maintain system operations and provide services to airlines, passengers, and other users. Here, $c$ denotes the unit operating cost, and $q$ represents the air traffic demand. The second term accounts for capacity (inventory) holding costs, which are expenses associated with maintaining capacity, including costs for energy and security. In this context, $c_h$ denotes the unit capacity holding cost, while $K$ represents the runway capacity. The third term represents delay costs, also known as congestion costs, which encompasses expenses related to flight delays, such as passenger compensation, aircraft re-routing, additional staffing, overtime pay, and lost revenue from cancellations.

According to \cite{sun2015stochastic}, delay costs are demands multiplied by the delay function:
\begin{equation}
\label{eq:delay_costs_def}
    D(q,K) = q F(\nu),
\end{equation}
where $\nu = q/K$ is the capacity utilization, and $F$ is the delay function. The delay function can take various mathematical forms depending on the practical context. \cite{sun2015stochastic} provides the following examples:
\begin{itemize}
    \item In highway networks, the following exponential form of the delay function is used:
    \begin{equation}
        F(\nu) = A(1+\alpha\nu^{\beta}),
    \end{equation}
    where $A$ is delay parameter and $\alpha, \beta$ are calibrated coefficients.
    \item In queueing systems, the delay function may take the following form from M/M/1 queues:
    \begin{equation}
        F(\nu) = A\left(\frac{\nu^2}{1-\nu}\right).
    \end{equation}
\end{itemize}

It is commonly assumed that demand cannot exceed capacity. \cite{balliauw2020expanding} states that the output is bounded by the designed capacity $K$, leading to the result $\nu = 1$ when $q \geq K$. \cite{sun2015stochastic} notes that when demand approaches capacity, the delay costs rise very sharply. Given the objective of minimizing total costs, demand naturally remains less than capacity. However, we argue that in certain cases --- especially at some of the world's busiest airports in metropolitan areas --- demand does exceed capacity. When this occurs, the airport is subject to strict schedule coordination
or slot control (\cite{jacquillat2018roadmap}), which leads to additional, but not unrealistically large costs. Therefore, potential overcapacity scenarios should be carefully considered in the capacity expansion planning process. To address this issue, \cite{taylor1984note} proposed a modified Davidson function for assignment problems in traffic networks. Inspired by \cite{taylor1984note}, we propose the following modification to the delay costs:
\begin{equation}
\label{eq:modified_davidson}
    \widetilde{D}(q, K) = \begin{cases}
        D(q, K), &\quad q \leq \mu K,\\
        D_{\mu} + K_{\mu} (q - \mu K), &\quad q > \mu K,
    \end{cases}
\end{equation}
In Eq. (\ref{eq:modified_davidson}), $\mu$ is a constant. $D_\mu$ is the value of the delay function at $q = \mu K$ and $K_\mu$ is the slope of the delay function at at the same point. These quantities are calculated as follows:
\begin{align}
    D_\mu &= D(q, \mu K),\\
    K_\mu &= \frac{\partial D(q, \mu K)}{\partial q} (\mu K).
\end{align}
We take Eq. (\ref{eq:delay_costs_def}) as the delay function. Calculation shows that
\begin{equation}
\resizebox{0.91\textwidth}{!}{%
$
 D(q, K) = Aq\left(1 + \alpha\left(\frac{q}{K}\right)^{\beta}\right), D_{\mu} = A\mu K(1 + \alpha \mu^{\beta}), K_{\mu} = A(1 + \alpha(\beta+1)\mu^{\beta}).
 $%
}
\end{equation}
Therefore, the modified delay costs can be calculated as follows:
\begin{equation}
\label{eq:modified_delay}
    \widetilde{D}(q, K) = \begin{cases}
        Aq\left(1 + \alpha\left(\frac{q}{K}\right)^{\beta}\right), &\quad q \leq \mu K,\\
        A(1 + \alpha(\beta+1)\mu^{\beta})q - A\alpha\beta\mu^{\beta + 1} K, &\quad q > \mu K.
    \end{cases}
\end{equation}
For simplicity, we set $\mu = 1$. Substituting Eq. (\ref{eq:modified_delay}) and $\mu = 1$ into Eq. (\ref{eq:system_cost_def}) yields
\begin{equation}
\label{eq:cost_function_def}
    C(q) = \begin{cases}
        cq + c_h K + Aq\left(1 + \alpha\left(\frac{q}{K}\right)^{\beta}\right), &\quad q \leq K,\\
        cq + c_h K + A(1 + \alpha(\beta+1))q - A\alpha\beta K, &\quad q > K.
        \end{cases}
\end{equation}

We consider two stages of airport runway system operations:
\begin{itemize}
    \item \textbf{Stage 1}: Before the runway capacity expansion, the airport operates at its initial capacity with calibrated parameter $K_0$.
    \item \textbf{Stage 2}: After the runway capacity expansion, the airport operates at an expanded capacity, reflected as change of parameter $K_0$ to $K_0+\Delta K$.
\end{itemize}

Substitute $K = K_0$ and $K = K_0 + \Delta K$ into Eq. (\ref{eq:cost_function_def}), we obtain the hourly system cost for Stage 1 and Stage 2, which are, respectively:
\begin{equation}
    \label{eq:C1}
    C_1(q) = \begin{cases}
        \dfrac{A\alpha}{K_0^{\beta}} q^{\beta + 1} + (A + c)q + c_h K_0, &\quad q \leq K_0,\\
        \left(A\left(1 + \alpha(\beta+1)\right) + c\right)q + (c_h - A\alpha\beta) K_0, &\quad q > K_0,
        \end{cases}
\end{equation}
\begin{equation}
\label{eq:C2}
\resizebox{0.91\textwidth}{!}{%
$
    C_2(q, \Delta K) = \begin{cases}
        \dfrac{A\alpha}{(K_0+\Delta K)^{\beta}} q^{\beta + 1} + (A + c)q + c_h (K_0+\Delta K), &\quad q \leq K_0+\Delta K,\\
        \left(A\left(1 + \alpha(\beta+1)\right) + c\right)q + (c_h - A\alpha\beta) (K_0+\Delta K), &\quad q > K_0+\Delta K,
        \end{cases} $%
        }
\end{equation}
where the subscript ``1" and ``2" denotes Stage 1 and Stage 2, respectively.

In addition to the costs associated with the airport runway system, another significant component is the one-time lump-sum investment required for the capacity expansion project, which includes both fixed and variable costs. Once the capacity expansion project is initiated, fixed costs are incurred regardless of the expansion size, while variable costs increase proportionally with the expansion size. The lump-sum investment cost can be expressed as follows:
\begin{equation}
    I = f + v\Delta K,
\end{equation}
where $f$ is the fixed costs and $v \Delta K$ is the variable costs. We assume that the capacity expansion project is completed instantly for ease of derivation, although in reality, a project could take several years. This assumption allows us to derive an elegant analytical solution, facilitating easier calculations without altering the key conclusions of results. Therefore, once the capacity expansion project is initiated, a one-time investment cost $I$ is incurred immediately. Notably, $f$ is incurred even if $\Delta K = 0$.

Among the components of airport runway system costs, only delay costs decrease with capacity expansion. Thus, delay cost savings represent a key benefit of expanding runway capacity. In contrast, operating costs and capacity holding costs increase with expansion to accommodate growing demand. Therefore, the trade-off between the savings from reduced delay costs and the additional costs of expansion must be carefully evaluated when making capacity expansion decisions.

\subsection{Deterministic dynamic model}
\label{sec:deterministic_dynamic_model}

In a deterministic setting, demand $Q_t$ grows over time according to the exponential growth function with a given growth rate $\eta$. $Q_t$ follows the following ordinary differential equation:
\begin{equation}
\label{eq:deterministic_Q_t_ode}
    \frac{d Q_t}{Q_t} = \eta dt.
\end{equation}
The demand function is derived by solving Eq. (\ref{eq:deterministic_Q_t_ode}):
\begin{equation}
\label{eq:deterministic_Q_t}
    Q_t = Q_0 e^{\eta t},
\end{equation}
where $Q_0$ is the initial demand at $t=0$. 

Following \cite{dangl1999investment, balliauw2020expanding}, we consider the cumulative cost discounted from the time of the capacity expansion project. Let the investment timing be $T$. If no capacity expansion project is undertaken, the cumulative system cost will only include operating costs, capacity holding costs, and delay costs. Since future costs are discounted from time $T$, the cumulative system cost is calculated as follows:
\begin{equation}
\label{eq:TC_1}
    TC_1(T) = \int_T^{\infty} N_p C_1(Q_t) e^{-\rho (t-T)} dt,
\end{equation}
where $N_p$ is the average runway operating hours per year, and $\rho$ is the discount rate.

Next, we analyze Stage 2, where the capacity expansion project is initiated at time $T$, and the airport runway system operates under the expanded capacity $K_0 + \Delta K$ from that point onward. The cumulative system cost, which includes the investment costs, is given as follows:
\begin{equation}
\label{eq:TC_2}
    TC_2(T, \Delta K) = \int_{T}^{\infty} N_p C_2(Q_t) e^{-\rho (t-T)} dt + (f + v \Delta K).
\end{equation}

We aim to maximize the difference between the cumulative costs given by Eqs. (\ref{eq:TC_1}) and (\ref{eq:TC_2}), specifically to maximize $TC_1(T) - TC_2(T, \Delta K)$. Let $q$ denote the demand at the investment timing $T$. We express the cumulative cost difference as a function of starting demand $q$ and expansion size $\Delta K$, denoted as $F(q, \Delta K)$. This function represents the cost savings accumulated from the investment timing $T$, given that the demand is $q$ at time $T$ and the expansion size is $\Delta K$. Therefore, as detailed calculations provided in \ref{sec: cal_detail_1}, we derive the following expression for $F(q, \Delta K)$:
\begin{equation}
\label{eq:deterministic_expected_cumulative_cost_difference}
    F(q, \Delta K)
    = TC_1(T) - TC_2(T, \Delta K) =\int_{0}^{\infty}\left[\Delta C(Q_t, \Delta K)  e^{-\rho t} dt\right].
\end{equation}
In Eq. (\ref{eq:deterministic_expected_cumulative_cost_difference}), $\Delta C(q, \Delta K)$ represents the annual cost difference function, defined as follows: 
\begin{equation}
\label{eq:deterministic_cost_difference_rate}
\resizebox{0.91\textwidth}{!}{%
$
\begin{aligned}
\Delta C(q, \Delta K) 
        &= N_p\left(C_1(q) - C_2(q, \Delta K)\right) - \rho(f + v \Delta K) \\
        &= \begin{cases}
          \Delta C_1(q, \Delta K) = \alpha_1 q^{\beta + 1} + \alpha_2, & \text{if } q < K_0, \\
          \Delta C_2(q, \Delta K) = \alpha_3 q^{\beta + 1} + \alpha_4 q + \alpha_5, & \text{if } K_0 \leq q < K_0 + \Delta K, \\
          \Delta C_3(q, \Delta K) = \alpha_6, & \text{if } q \geq K_0 + \Delta K,
          \end{cases}
\end{aligned}
$%
}
\end{equation}
where $C_1$ and $C_2$ are defined in Eqs. (\ref{eq:C1}) and (\ref{eq:C2}). In Eq. (\ref{eq:deterministic_cost_difference_rate}), $\alpha_1 = N_pA\alpha(K_0^{-\beta} - (K_0+\Delta K)^{-\beta}), \alpha_2 = - (N_p c_h + \rho v) \Delta K - \rho f, \alpha_3 = - N_pA\alpha(K_0+\Delta K)^{-\beta}, \alpha_4 = N_p A\alpha(\beta + 1), \alpha_5 = - (N_p c_h + \rho v) \Delta K - N_p A\alpha\beta K_0 - \rho f$, and $\alpha_6 = - (N_p (c_h-A\alpha\beta) + \rho v) \Delta K - \rho f$.

\begin{theorem}
\label{thm: F_deterministic}
    The cost saving function for the deterministic dynamic model with starting demand $q$ and investment size $\Delta K$ is given by
    \begin{equation}
        F(q, \Delta K) = A_1q^{\beta + 1} + A_2q^{\frac{\rho}{\eta}}+A_3,
    \end{equation}
    where $A_1, A_2, A_3$ are functions of $\Delta K$.
\end{theorem}

\begin{proof} 
See \ref{sec: proof_F_deterministic}.
\end{proof}

For traditional NPV methods, the capacity expansion project is initiated when it becomes profitable, i.e., when $F(q,\Delta K) \geq 0$ (\cite{dixit1994investment}). To determine the trigger demand, we first need to find the optimal $\Delta K$ for each $q$ that maximizes $F(q,\Delta K)$. Then, the trigger demand is given by:
\begin{equation}
    q_{NPV} = \min\{q: \max_{\Delta K}~\{F(q,\Delta K)\} \geq 0\}.
\end{equation}
The expansion size at the trigger demand $q_{NPV}$ is given by:
\begin{equation}
    \Delta K_{NPV} = \argmax_{\Delta K} ~\{F(q_{NPV},\Delta K)\}.
\end{equation}

It is crucial to note that we cannot directly maximize $F(q,\Delta K)$ to find the trigger demand and the corresponding expansion size, as demand is a dynamically changing process, and the cost savings at the present time cannot be compared to that at a future time. Readers may refer to \cite{dixit1994investment} for further insights.

\subsection{Stochastic dynamic model}
\label{sec:stochastic_dynamic_model}
\subsubsection{Uncertain demand}
Real-world air traffic demand is inherently uncertain, exhibiting fluctuations similar to those of the stock market. This uncertainty is reflected in data from a reliable open platform that tracks real-time flight activity worldwide, as shown in Figure \ref{fig:covid-19} (\cite{flightradar24}). For example, Chicago O'Hare International Airport (ORD), one of the busiest airports in the United States, experienced 18,369 delayed arrivals and departures in June 2019 alone, with a delay rate of 31.98\% (\cite{bts}). Another source of uncertainty comes from rare events that cause sudden drops in demand. A significant example is the sharp decline in the number of flights following the unexpected onset of the COVID-19 pandemic, which resulted in a dramatic 63.5\% drop in passenger air traffic in 2020. 

\begin{figure}[ht]
\centering
\includegraphics[width=10cm]{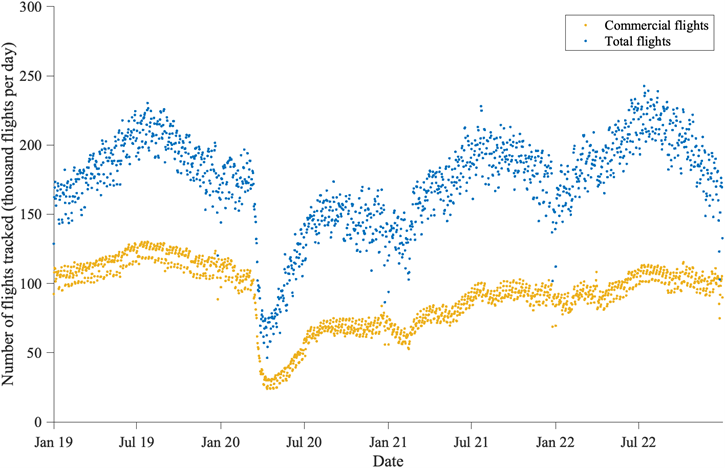}
\caption{\label{fig:covid-19}Number of flights tracked during the COVID-19 pandemic.}
\end{figure}

Expanding an airport's runway capacity requires a significant financial investment. Failing to account for demand uncertainty may result in two critical issues. On the one hand, if an airport overestimates future demand and expands capacity unnecessarily, it may lead to underutilized infrastructure and wasted resources, resulting in financial losses. On the other hand, if demand is underestimated, the airport might experience congestion, flight delays, and a decline in service quality, all of which can negatively impact both revenue and reputation. By accounting for demand uncertainty, airport authorities can optimize resource allocation and tailor expansion plans to align more accurately with projected demand, ensuring that expansion meets actual needs. 

To better model uncertainty in demand growth, we adopt a jump-diffusion process, as outlined in \cite{guo2023investment}, to capture demand dynamics. This process combines elements of a Geometric Brownian Motion (GBM) and a Poisson process, effectively representing the stochastic nature of air traffic demand. The stochastic differential equation for demand is expressed as follows:
\begin{equation}
\label{eq:sde_basic}
    \frac{d Q(t)}{Q(t)} = \eta dt + \sigma d w(t) + d \sum_{i=1}^{N_t} Z_i,
\end{equation}
where $\eta$ represents the demand growth rate, and $\sigma$ refers to the volatility rate. Here, $\eta$ reflects the expected change in demand, while $\sigma$ captures the degree of random fluctuations. The term $dw(t)$ represents an increment of Brownian motion. By definition, $dw(t)=\epsilon_t\sqrt{t}$ , where $\epsilon_t$ follows the standard normal distribution, i.e., $\epsilon_t \sim N(0,1)$. In the jump component $\sum_{i=1}^{N_t} Z_i$ , $N_t$ is modeled by a Poisson process with intensity $\lambda$. $Z_i$ is the magnitude of the $i$-th jump (e.g., a percentage change such as a 10\% decrease) and is independent of GBM. Therefore, the jump component in Eq. (\ref{eq:sde_basic}), i.e., $d\sum_{i=1}^{N_t} Z_i$ , equals $Z_i$ if a jump occurs, with probability $\lambda dt$; otherwise, it equals 0 with probability $(1-\lambda dt)$. The jump sizes $Z_i$ are independently and identically distributed, following the same distribution as $Z$. $Z$ can be a constant or random variable.

\subsubsection{Optimal capacity expansion strategy}
\label{sec:Optimal capacity expansion strategy}
Similar to the deterministic dynamic model in Section \ref{sec:deterministic_dynamic_model}, we aim to maximize the expectation of the cumulative cost difference between Eq. (\ref{eq:TC_1}) and Eq. (\ref{eq:TC_2}), specifically maximizing $\mathbb{E}[TC_1(\tau) - TC_2(\tau, \Delta K)]$. In the deterministic dynamic model, the demand threshold corresponds to a fixed timing, whereas in the stochastic case, it corresponds to a random timing $\tau$. Since $\tau$ is a random variable, expectation calculation is required. We seek to find a demand threshold, denoted as $q^*$, which is the optimal trigger demand for initiating the capacity expansion project. 

Let the demand at the investment timing be denoted as $q$. We write the expected cumulative cost difference as a function of $q$ and the expansion size $\Delta K$, represented as $F (q, \Delta K)$:
\begin{equation}
\label{eq:expected_cumulative_cost_difference}
\resizebox{0.91\textwidth}{!}{%
$
    F (q, \Delta K)
    = \mathbb{E}_{|Q(\tau)=q}\left[TC_1(\tau) - TC_2(\tau, \Delta K)\right]= \mathbb{E}_{|Q(0)=q}\int_{0}^{\infty}\left[\Delta C(Q(t), \Delta K)  e^{-\rho t} dt\right],
$%
}
\end{equation}
where the cost difference function $\Delta C(q, \Delta K)$ is defined in Eq. (\ref{eq:deterministic_cost_difference_rate}). The detailed calculations can be found in \ref{sec: cal_detail_2}.

We consider the following decision problem: given the demand level $q$, should the airport authority initiate the capacity expansion project immediately or wait? Let $V(q)$ be the option function. If the demand at time $t$ is larger than the demand threshold at the investment time $\tau$, the decision has been made, so $V(q)$ satisfies the following equation:
\begin{equation}
\label{eq:option_post}
    V(q | Q(t) = q \geq  Q(\tau)) = \max_{\Delta K}\{F(q, \Delta K)\}.
\end{equation}
If the decision is not made, $V(q)$ satisfies the following dynamic programming formula:
\begin{equation}
\label{eq:option_pred}
    V(q | Q(t) = q <  Q(\tau)) = e^{-\rho dt} \mathbb{E} [V(q) + d V(q)].
\end{equation}
Combining Eqs. (\ref{eq:option_post}) and (\ref{eq:option_pred}), and considering that for any time $t$, there exists the option to initiate the capacity expansion project immediately or wait (i.e., ``exercise the option" or ``keep the option"), the dynamic programming formula for $V(q)$ is:
\begin{equation}
\label{eq:option_dp}
    V(q) = \max\left\{e^{-\rho dt} \mathbb{E} [V(q) + d V(q)],\quad \max_{\Delta K} \{F(q, \Delta K)\}\right\}.
\end{equation}

\begin{theorem}
\label{thm:F_function}
The expected cost saving function for the stochastic dynamic model with starting demand $q$ and investment size $\Delta K$ is given by
    \begin{equation}
    F(q, \Delta K)_{q \in R_i} = F_i(q, \Delta K) = A_{i, 1} q^{b_1} + A_{i, 2} q^{b_2} + \overline{F_i}(q, \Delta K),\quad i = 1,2,3,
    \end{equation}
    where $R_1 = \{q: q < K_0\}, R_2 = \{q: K_0 \leq q < K_0 + \Delta K\}, R_3 = \{q: q \geq K_0 + \Delta K\}$. $\overline{F_i}$ is given by
    \begin{equation}
    \label{eq:ode_particular_sol}
    \resizebox{0.91\textwidth}{!}{%
$
    \begin{aligned}
        \overline{F_1}(q, \Delta K) &= -\frac{\alpha_1}{\sigma^2\beta(\beta+1)/2 + \eta(\beta+1) + \lambda \mathbb{E}_Z[(1+Z)^{\beta + 1}] - (\lambda + \rho)} q^{\beta + 1} + \frac{\alpha_2}{\rho},\\
        \overline{F_2}(q, \Delta K) &=-\frac{\alpha_3}{\sigma^2\beta(\beta+1)/2 + \eta(\beta+1) + \lambda \mathbb{E}_Z[(1+Z)^{\beta + 1}] - (\lambda + \rho)} q^{\beta + 1} -\frac{\alpha_4}{\eta + \lambda \mathbb{E}_Z[Z] - \rho} q + \frac{\alpha_5}{\rho},  \\
        \overline{F_3}(q, \Delta K) &= \frac{\alpha_6}{\rho},
    \end{aligned}
    $%
    }
    \end{equation}
    $b_1 > 1, b_2 < 0$ are two solutions of the following equation:
    \begin{equation}
    \label{eq:b_first_appear}
        \varphi(b) := \frac{\sigma^2}{2}b(b-1) + \eta b + \lambda \mathbb{E}_Z (1+Z)^{b} -(\rho+\lambda) = 0,
    \end{equation}
    and $A_{i,j} (i\in\{1,2,3\}, j\in\{1,2\})$ are functions of $\Delta K$.
\end{theorem}

\begin{proof}
See \ref{sec: proof_F_stochastic}.
\end{proof}

Given $F(q, \Delta K)$, one can derive the optimal expansion size $\Delta K$ for any given $q$ by maximizing $F(q, \Delta K)$. The function value of $\max_{\Delta K} F(q, \Delta K)$ can be treated as the ``termination payoff" when ``exercising the option". Having solved the inner maximization of Eq. (\ref{eq:option_dp}), we now turn to the outer maximization to determine the investment timing --- specifically, at which level of $q$ the investment decision should be made. If the investment is not made at demand $q$, the option function $V(q)$ satisfies the following theorem.

\begin{theorem}
\label{thm: solving_V}
    The option function $V(q)$ is given by
    \begin{equation}
    \label{eq:V_thm_eq}
        V(q) = \overline{A_1} q^{b_1},
    \end{equation}
    where
    \begin{equation}
    \label{eq:q*V_thm_eq_A1}
        \overline{A_1} = \max_{q}\left\{\frac{\max_{\Delta K} F(q, \Delta K)}{q^{b_1}}\right\},
    \end{equation}
    $b_1$ is given by Theorem \ref{thm:F_function}.
    \begin{proof}
        See \ref{sec: proof_trigger}.
    \end{proof}
\end{theorem}

To determine the optimal trigger demand $q^*$ and the expansion size $\Delta K^*$, we apply the ``value-matching" and ``smooth-pasting" conditions (\cite{dixit1994investment}), with additional details provided in Appendix \ref{sec: proof_trigger}. The results are summarized in Theorem \ref{thm:trigger} below.

\begin{theorem}
\label{thm:trigger}
    The trigger demand for the stochastic dynamic model is given by
    \begin{equation}
        q^* = \argmax\limits_{q}~ \left\{\frac{\max_{\Delta K} F(q, \Delta K)}{q^{b_1}}\right\},
    \end{equation}
    where $F(q, \Delta K)$ and $b_1$ is given by Theorem \ref{thm:F_function}.

    The expansion size at the trigger demand $q^*$ is given by
    \begin{equation}
    \label{eq:expansion_size_stoc}
        \Delta K^* = \argmax\limits_{\Delta K}~ \left\{F(q^*, \Delta K)\right\}.
    \end{equation}
\begin{proof}
    See \ref{sec: proof_trigger}.
\end{proof}
\end{theorem}

\section{Model applications}
\label{sec: Model Applications}
The baseline numerical inputs for both deterministic and stochastic models are provided in Table \ref{tab:variable_definition}.

\subsection{Deterministic analyses}

In this subsection, we calculate the optimal investment decision for airport runway capacity expansion using the deterministic dynamic model from Section \ref{sec:deterministic_dynamic_model}. Considering that the runway capacity is discrete, we set a maximum of 40 operations per hour per runway (according to \cite{blumstein1959landing, newell1979airport}). Thus, the expansion size, $\Delta K$, must be a multiple of 40. Initially, the airport has one runway, with $K_0 = 40$. Solving the model yields the optimal trigger of $q_{NPV} = 5.4$ operations per hour and an expansion size of $\Delta K_{NPV} = 40$ operations per hour, indicating that one additional runway is required.

To further investigate the relationship between $Q_0$ and $\Delta K$, we consider the following investment question: Given a decision to invest immediately with the current demand level at $Q_0$, what is the optimal choice for capacity expansion? The results are shown in Figure \ref{fig:determ_q_delta_K}, where the dashed yellow line represents the expansion size $\Delta K$, and the green solid line represents the cost saving function $F$ given the current demand level and the corresponding optimal expansion size. For notation convenience here and in the following text, we denote $F(q) = \max_{\Delta K}~\{F(q, \Delta K)\}$. For any $Q_0$, the optimal expansion size is given by $\Delta K (Q_0) = \argmax_{\Delta K}~\{F(Q_0, \Delta K)\}$, while the function value of $F$ is given by $F(Q_0) = F(Q_0, \Delta K (Q_0)) = \max_{\Delta K}~\{F(Q_0, \Delta K)\}$.
The figure shows
that a higher value of $Q_0$ implies a greater expansion size. When $Q_0$ is small, $\Delta K = 0$, indicating that the optimal choice is to ``keep the option" without expanding runway capacity. As $Q_0$ increases, the optimal choice shifts to adding one, two, or more runways. Thus, higher initial demand levels require additional runways to accommodate flights efficiently.

\begin{figure}[ht]
\centering
\includegraphics[width=13cm]{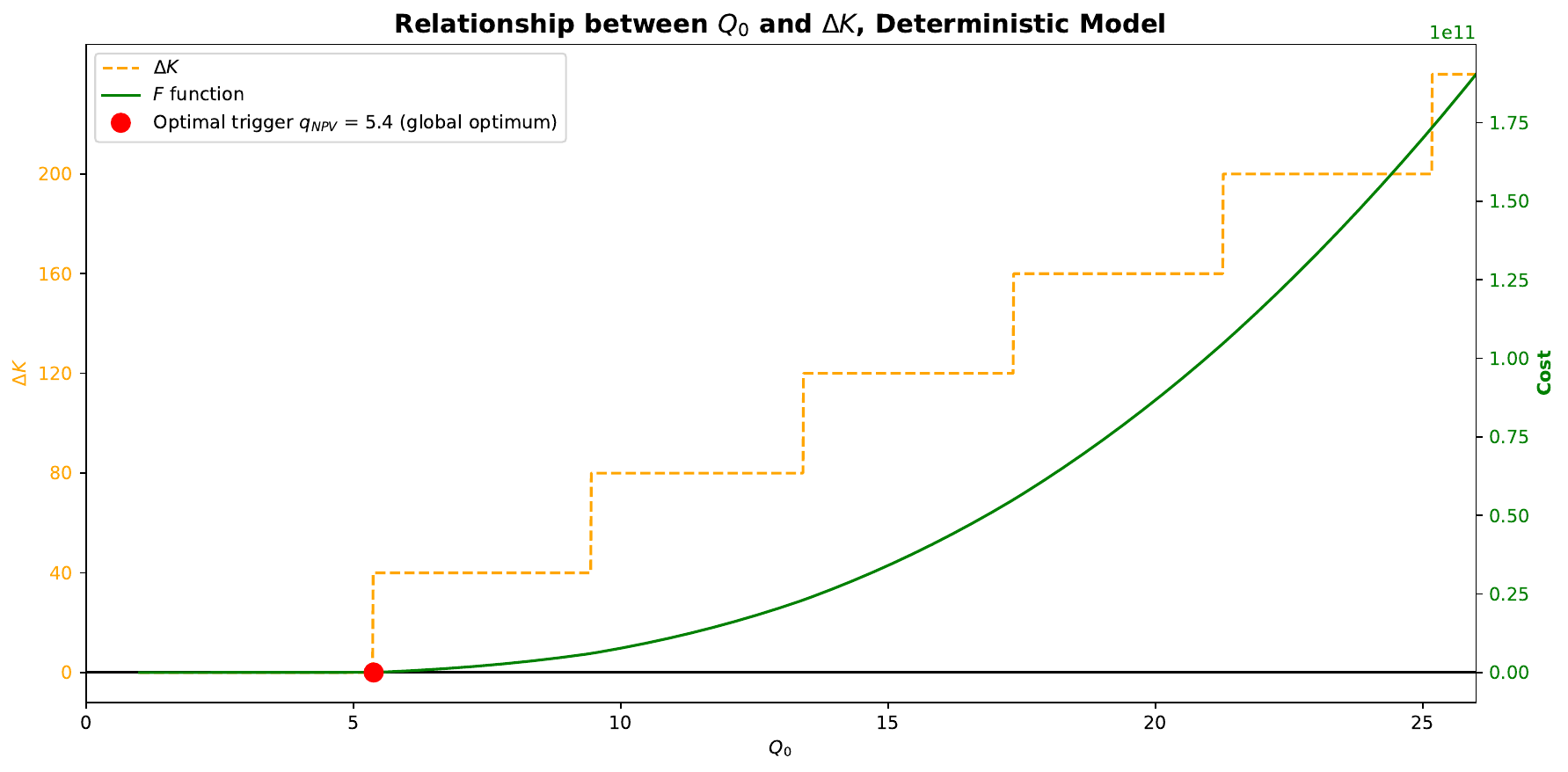}
\caption{\label{fig:determ_q_delta_K}Relationship between $Q_0$ and $\Delta K$ for the deterministic model.}
\end{figure}

Moreover, if the expansion size $\Delta K$ is fixed, the trigger demand for the deterministic model is given in Figure \ref{fig:determ_delta_K_q}. From the figure, it is evident that if the expansion size increases, the corresponding trigger demand also increases.

\begin{figure}[ht]
\centering
\includegraphics[width=13cm]{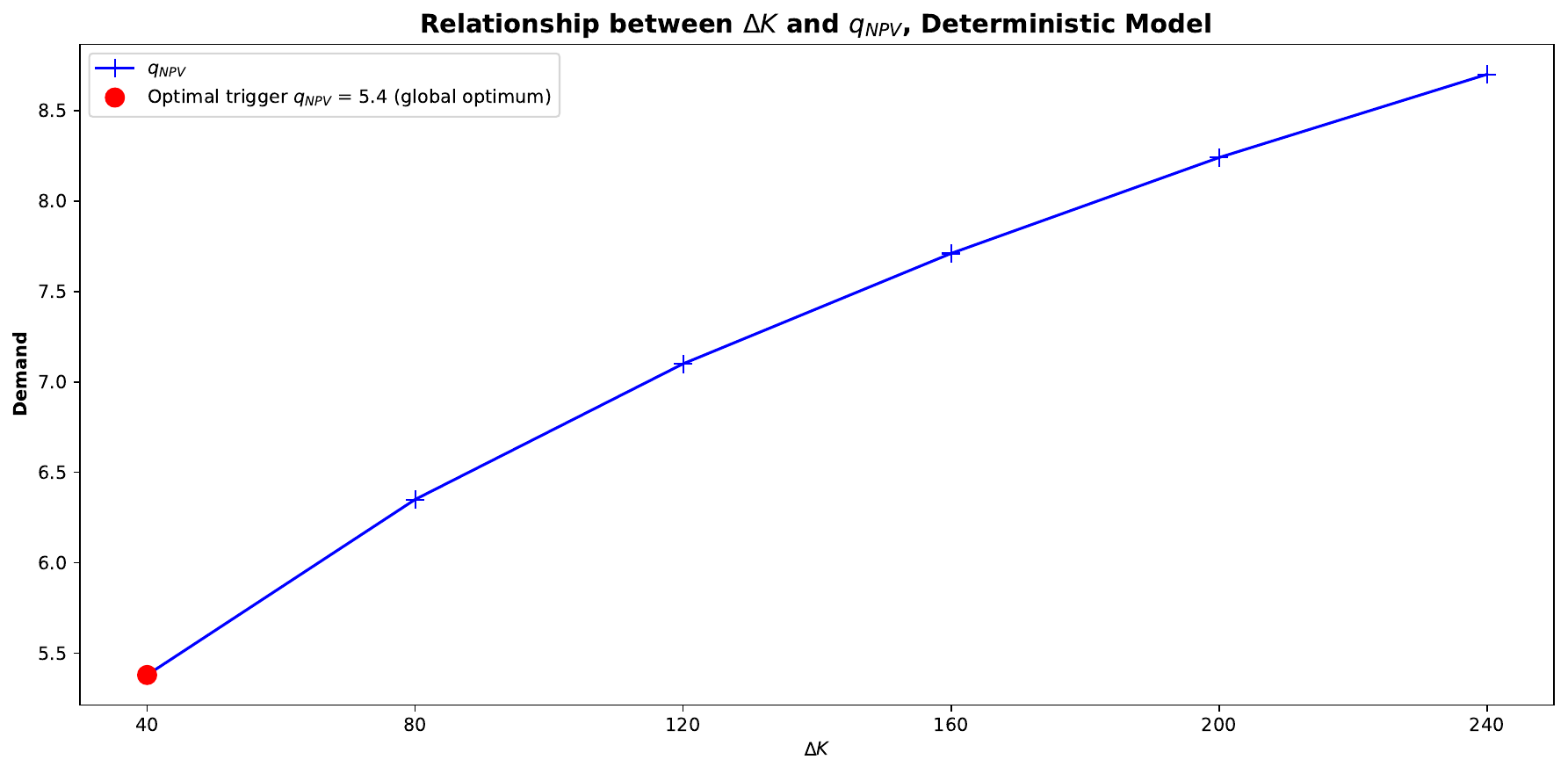}
\caption{\label{fig:determ_delta_K_q}Relationship between $\Delta K$ and $q_{NPV}$ for the deterministic model.}
\end{figure}

\subsection{Stochastic analyses}

It is important to understand that the stochastic process defined by Eq. (\ref{eq:sde_basic}) encompasses an infinite number of potential demand trajectories, resulting from inherent uncertainty. This variability introduces additional complexity into capacity expansion planning. Figure \ref{fig:trajectory} illustrates how demand generally follows an exponential growth trend, influenced by stochastic fluctuations and potential drops.

\begin{figure}[ht]
\centering
\includegraphics[width=13cm]{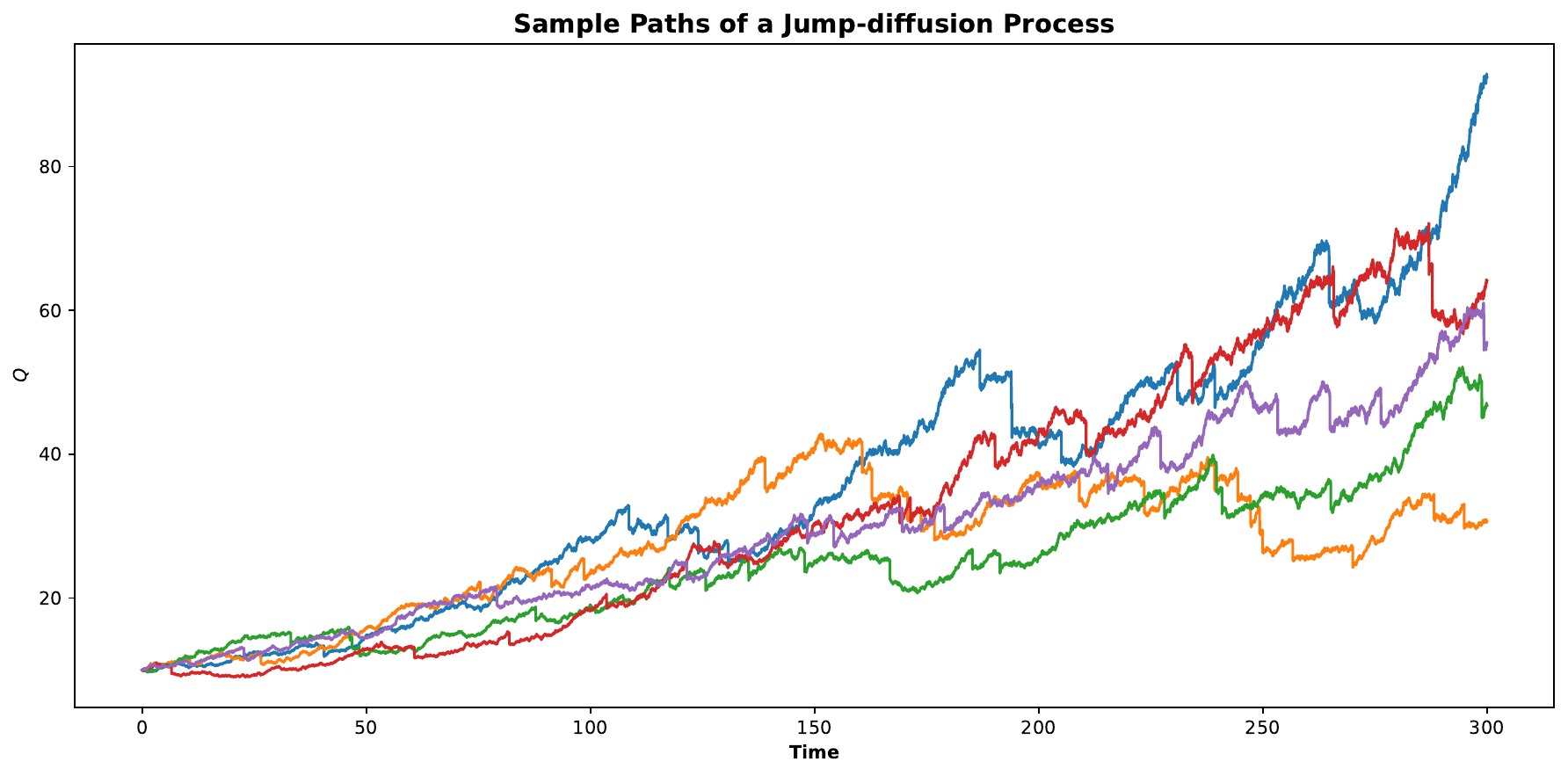}
\caption{\label{fig:trajectory}Sample paths of a jump-diffusion process.}
\end{figure}

The jump magnitude $Z$ may follow different probability distributions. For simplicity, we consider the jump process with a constant magnitude for each downward jump, i.e., $Z_i = - 10\%, i = 1,\cdots, N_t$. The values of all other parameters are present in Table \ref{tab:variable_definition}. The solution for the stochastic model follows from Theorem \ref{thm:trigger}. Figure \ref{fig:solve_V_F_stochastic} displays both the $F$ and $V$ functions, with the red star indicating the trigger demand $q^*$, where the $F$ and $V$ functions intersect (value-matching) and have matching slopes (smooth-pasting). The trigger demand is $q^* = 13.7$ operations per hour, with a corresponding investment size of $\Delta K^* = 80$ operations per hour. Compared to $q_{NPV}$, the real options approach advocates for postponing the investment decision: the investment is only made if there is sufficient evidence that current demand growth is driven by a sustained upward trend rather than temporary fluctuations. Additionally, a larger expansion size is required to accommodate the higher trigger demand. The stochastic dynamic model is more realistic, as it incorporates two layers of uncertainty that can lead to fluctuations, as represented by the jump-diffusion process: first, volatility introduced by Brownian motion, and second, potential downward jumps driven by the Poisson process. These uncertainties necessitate a more conservative capacity expansion decision.

\begin{figure}[ht]
\centering
\includegraphics[width=13cm]{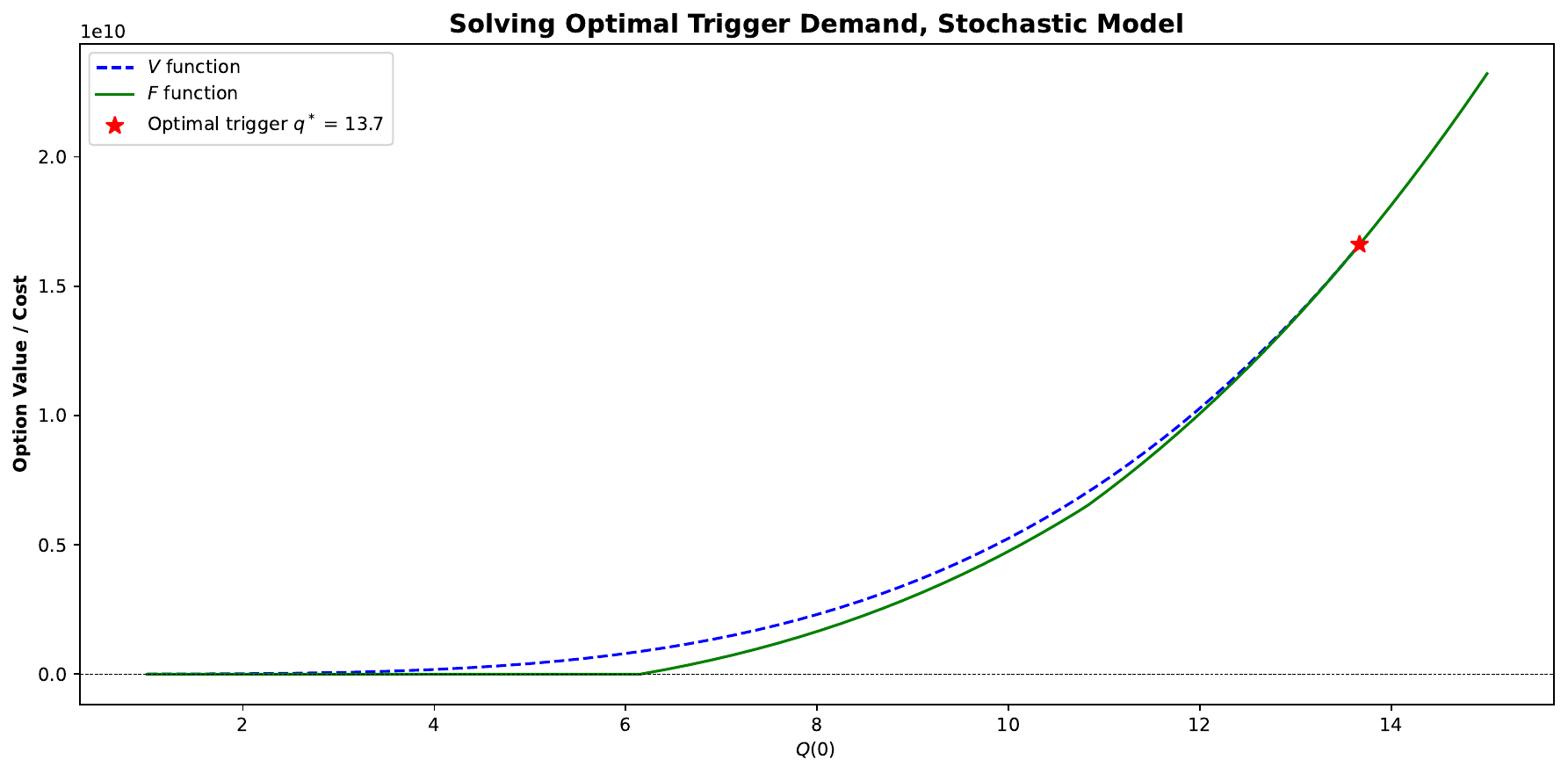}
\caption{\label{fig:solve_V_F_stochastic}Solving optimal trigger demand for the stochastic model.}
\end{figure}

In a manner akin to the deterministic analyses, Figure \ref{fig:stocha_q_delta_K} presents the capacity expansion choice given the current demand level at $Q(0)$. The trend of the relationship between $Q(0)$ and $\Delta K$ is similar to the deterministic model, where $\Delta K$ is a stepwise function of the initial demand. As $Q(0)$ increases, $\Delta K$ increases as well, indicating the need for more airport runways. It is also evident in the figure that the real options model postpones the investment decision: in terms of expectation, the $F$ function becomes positive at a smaller demand compared with $q^*$. However, considering the uncertainties inherent in demand growth, along with the risk-averse attitude of decision makers, it is not optimal to make the capacity expansion decision once the expected cost saving function becomes positive. Figure \ref{fig:stocha_delta_K_q} presents the trigger demand given the expansion size $\Delta K$. Similar to the deterministic model, an increase in the expansion size leads to a corresponding increase in the trigger demand.

\begin{figure}[ht]
\centering
\includegraphics[width=13cm]{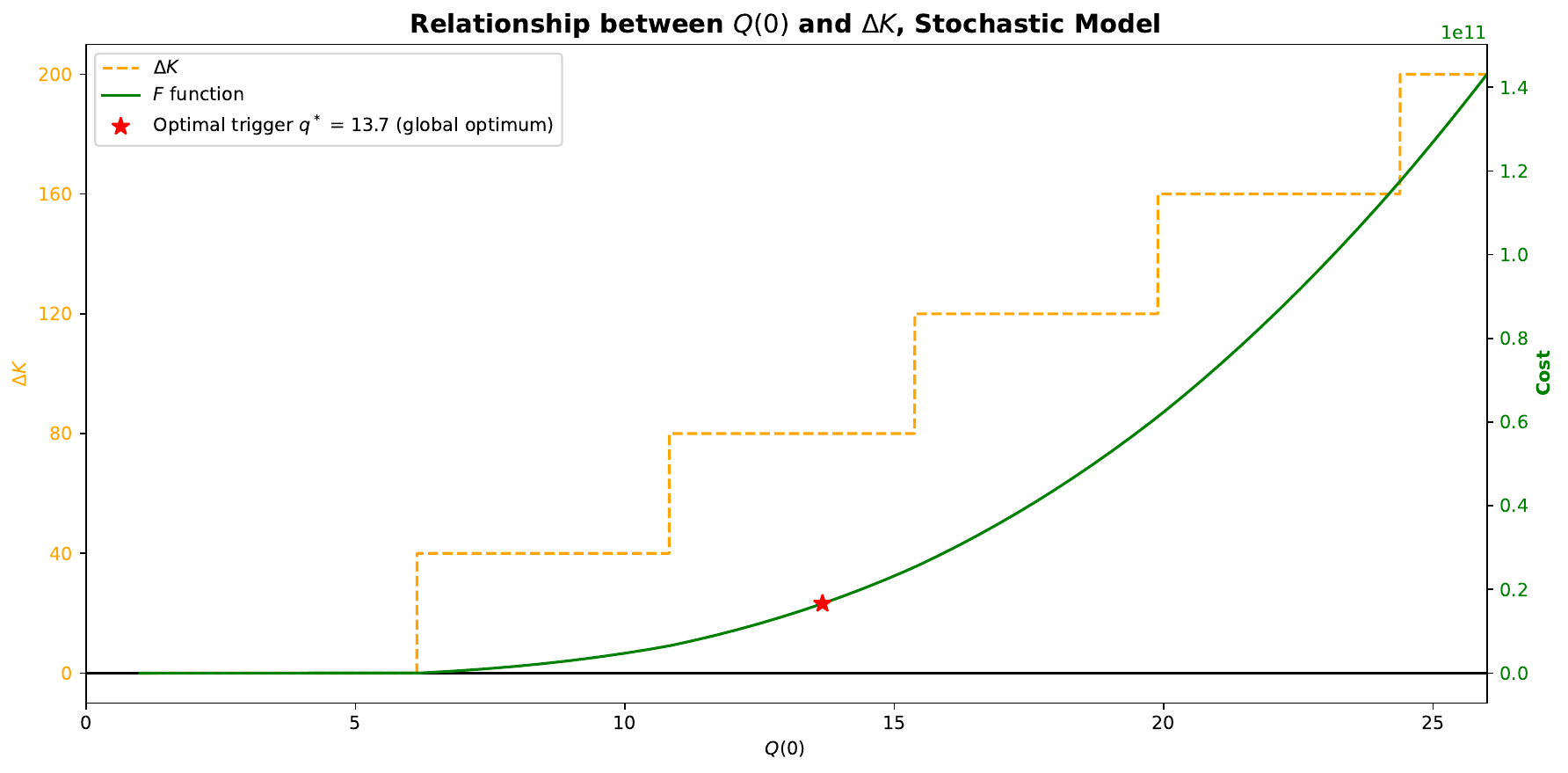}
\caption{\label{fig:stocha_q_delta_K}Relationship between $Q(0)$ and $\Delta K$ for the stochastic model.}
\end{figure}

\begin{figure}[ht]
\centering
\includegraphics[width=13cm]{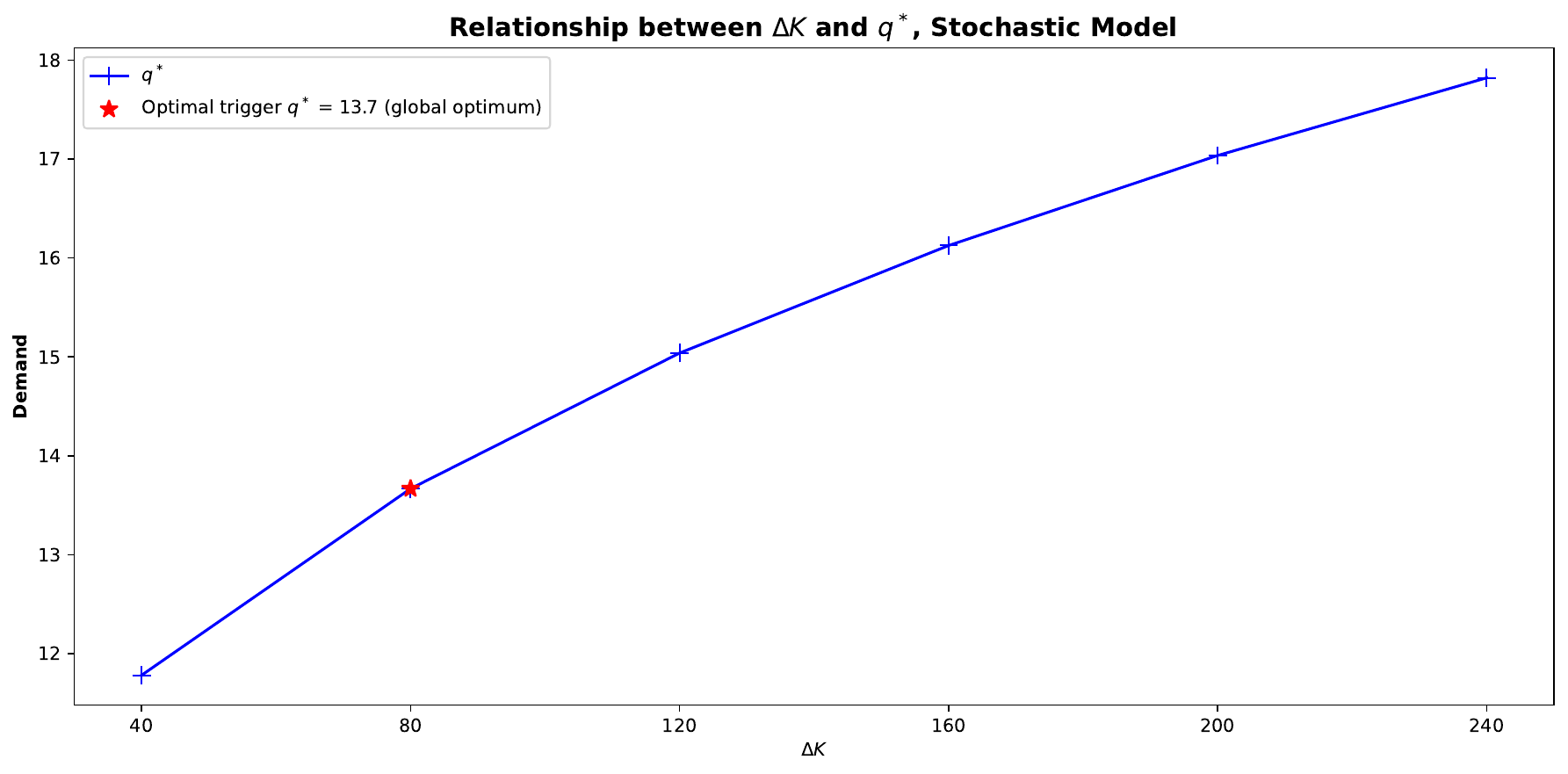}
\caption{\label{fig:stocha_delta_K_q}Relationship between $\Delta K$ and $q^*$ for the stochastic model.}
\end{figure}

\subsection{Sensitivity analyses}

In this subsection, we conduct sensitivity analyses of several key parameters while keeping other variables unchanged for the stochastic dynamic model. The baseline numerical inputs are provided in Table \ref{tab:variable_definition}. First, we vary the initial capacity $K_0$, which represents scenarios where airports initially have different numbers of runways. The results are shown in Figure \ref{fig:sensitivity_K0}. As $K_0$ increases, both the trigger demand $q^*$ and the expansion size $\Delta K^*$ increase. This is intuitive because the congestion effect plays a significant role only when demand approaches the capacity. Therefore, a larger initial capacity suggests a higher trigger demand. In order to accommodate the fast-growing demand, a larger expansion size is also required.

\begin{figure}[ht]
\centering
\includegraphics[width=13cm]{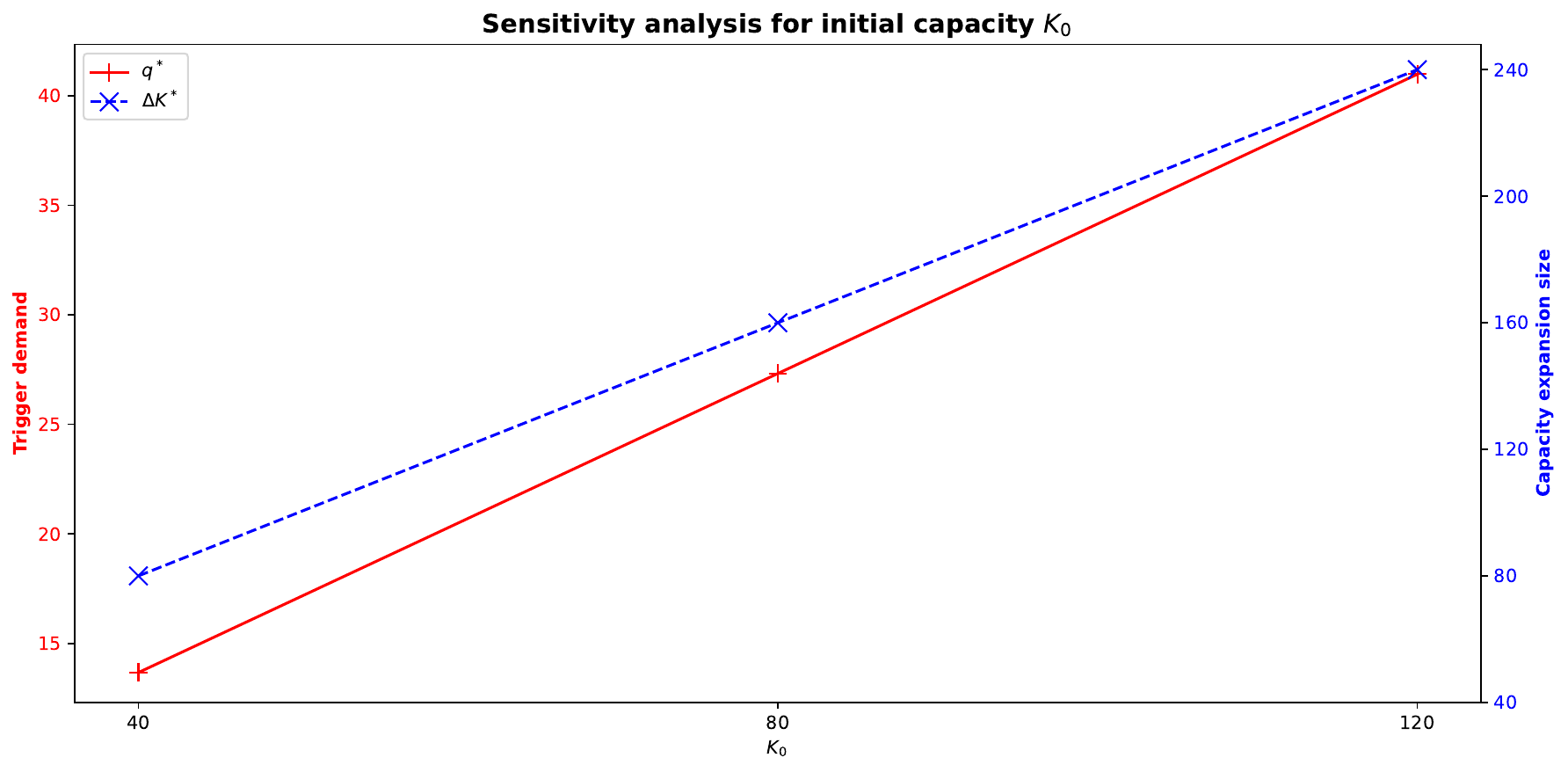}
\caption{\label{fig:sensitivity_K0}Sensitivity analysis for initial capacity $K_0$.}
\end{figure}

Figure~\ref{fig:sensitivity_eta} shows the optimal decisions for different demand growth rates $\eta$. As $\eta$ increases, the optimal expansion size $\Delta K^*$ also increases. However, the relationship between $\eta$ and $q^*$ is more complex. With the same expansion size $\Delta K^*$, a larger $\eta$ indicates a lower trigger demand $q^*$. At the point when $\Delta K^*$ jumps to the next level, an upward jump in $q^*$ is also required, because only a larger threshold can accommodate a larger expansion size. This relationship, although consistent with the literature, presents its distinct features compared to studies assuming that the expansion size is a continuous variable, such as \cite{guo2018time, guo2023investment}. This provides valuable insights for decision-makers: if the expansion size is predetermined, anticipating a higher demand growth rate can justify advancing the investment decision; however, when it becomes clear that a larger expansion size is necessary, the decision should be postponed until the demand grows sufficiently to accommodate that additional capacity.

\begin{figure}[ht]
\centering
\includegraphics[width=13cm]{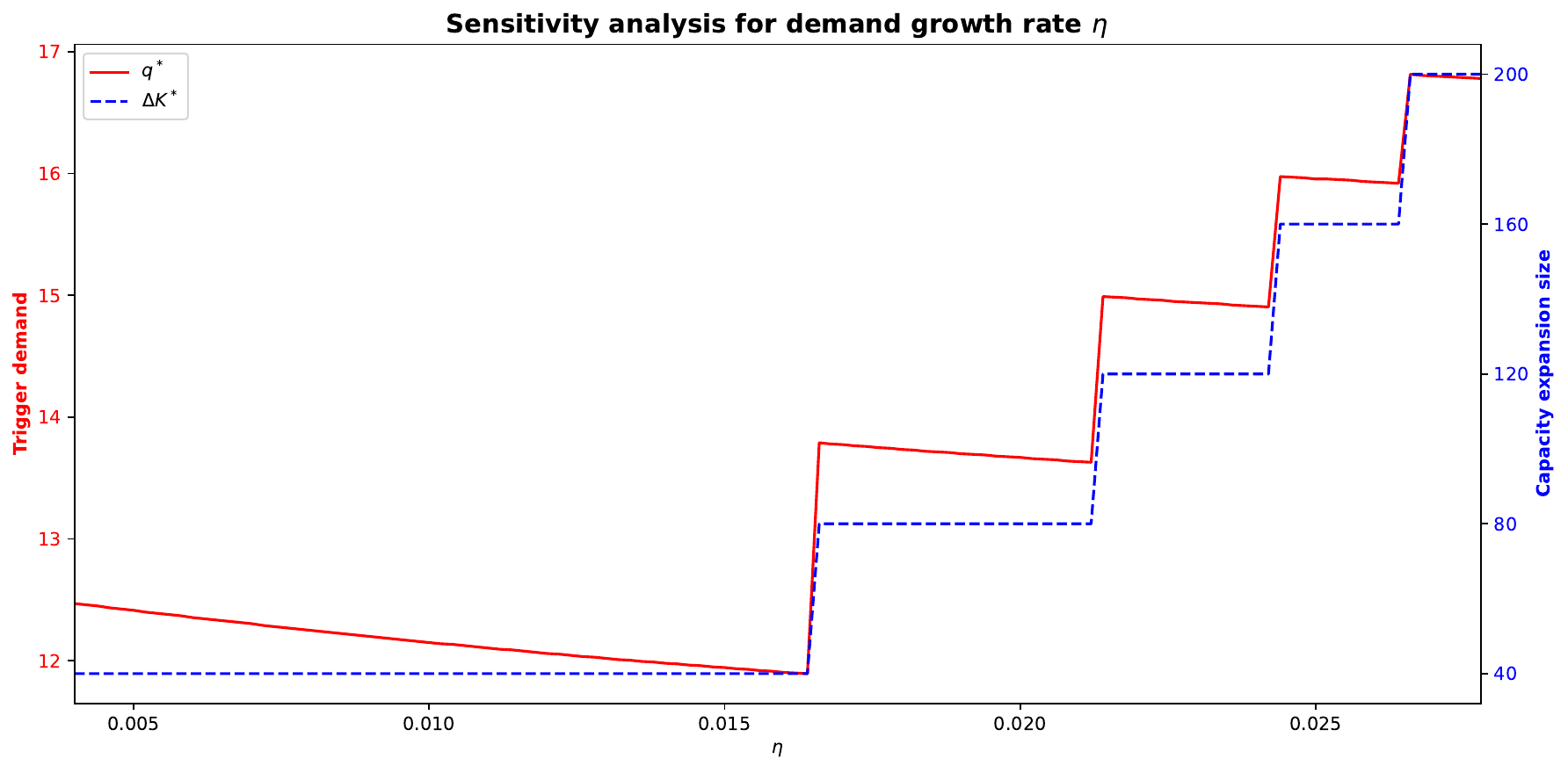}
\caption{\label{fig:sensitivity_eta}Sensitivity analysis for demand growth rate $\eta$.}
\end{figure}

Figure \ref{fig:sensitivity_sigma} presents the optimal decisions for different volatility rates $\sigma$. As $\sigma$ increases, both the optimal expansion size $\Delta K^*$ and the trigger demand $q^*$ also increase. A higher volatility rate indicates greater uncertainty in demand growth, necessitating a larger trigger demand to ensure the profitability of the capacity expansion decision. Furthermore, when the expansion size jumps --- indicating that an additional runway is required --- the trigger demand $q^*$ also experiences a jump. This occurs because a larger expansion size necessitates a higher demand level to justify its effectiveness.

\begin{figure}[ht]
\centering
\includegraphics[width=13cm]{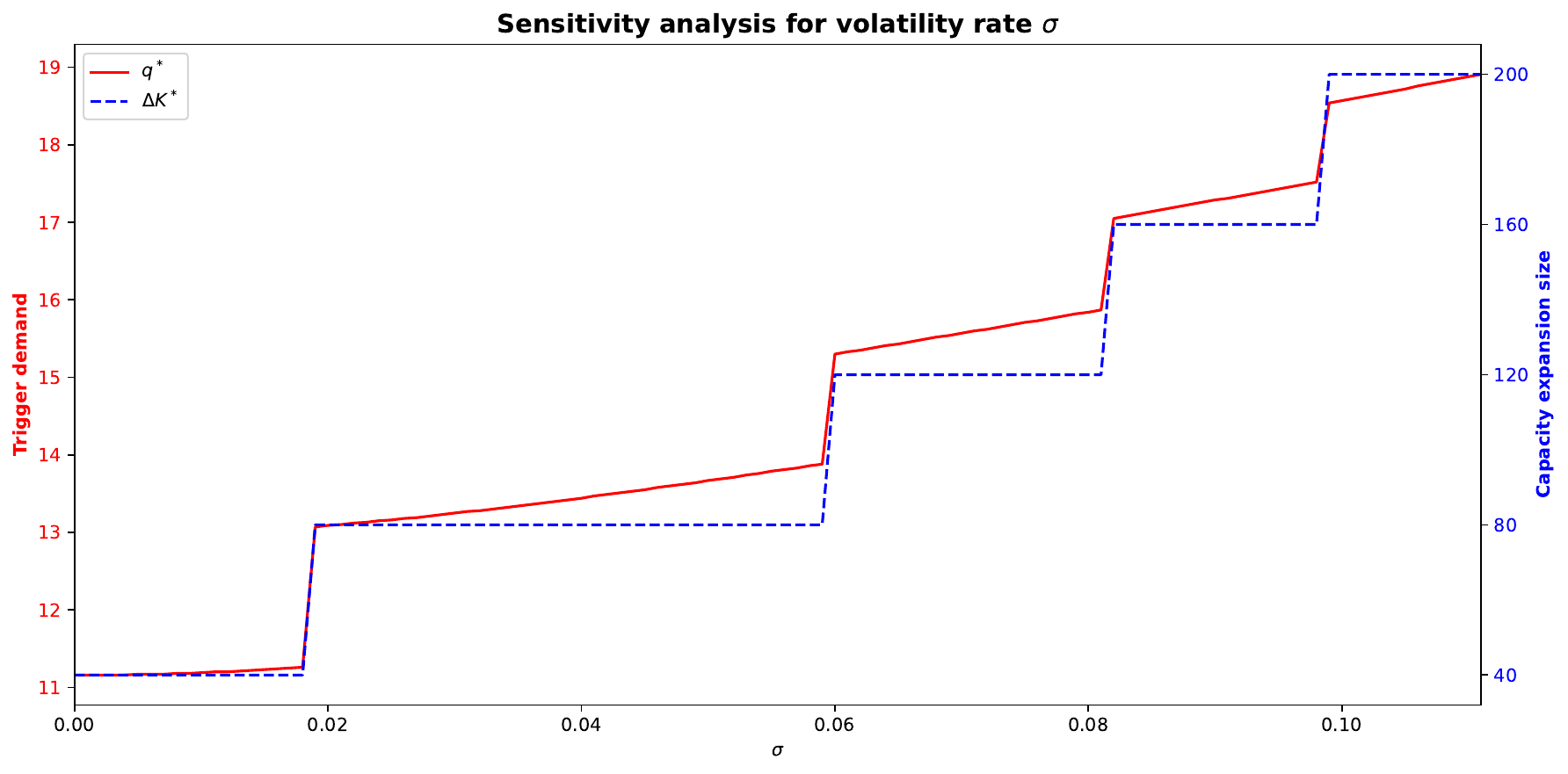}
\caption{\label{fig:sensitivity_sigma}Sensitivity analysis for volatility rate $\sigma$.}
\end{figure}

Figure \ref{fig:sensitivity_jump_size} illustrates the optimal decisions for varying jump sizes. A larger jump size corresponds to a smaller expansion size, as increased uncertainty necessitates a more conservative investment. Regarding the trigger demand $q^*$, with a fixed expansion size, it increases with a larger jump size, as more evidence is required for the investment decision, resulting in a postponed expansion. However, at the point when the expansion size drops, the trigger demand drops as well. A similar trend is observed when altering the jump probability $\lambda$, as shown in Figure \ref{fig:sensitivity_jump_prob}. 

\begin{figure}[ht]
    \centering
    \subfigure[Sensitivity analysis for jump size.]{%
    \resizebox{6.7cm}{!}{\label{fig:sensitivity_jump_size}\includegraphics{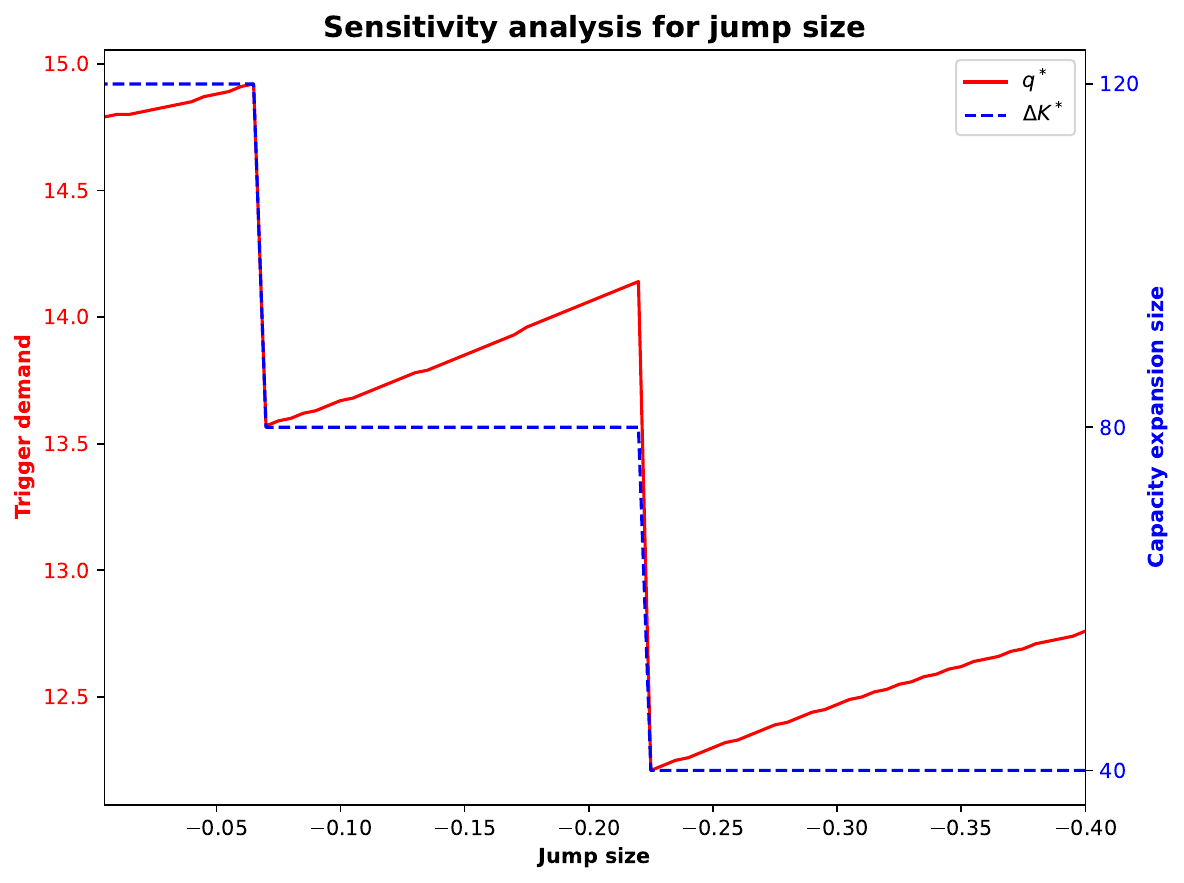}}}\hspace{5pt}
    \subfigure[Sensitivity analysis for jump probability.]{%
    \resizebox{6.7cm}{!}{\label{fig:sensitivity_jump_prob}\includegraphics{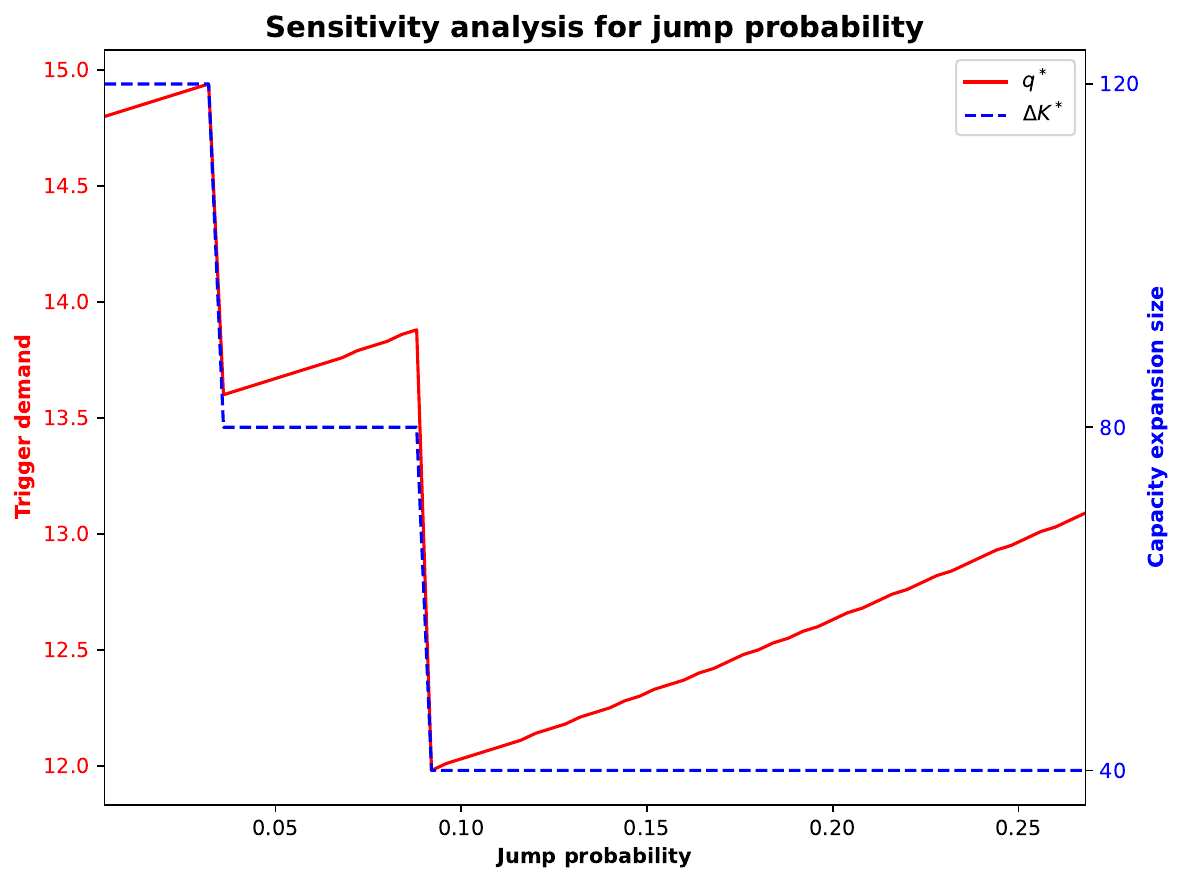}}}
    \caption{Sensitivity analysis for random jumps.}
    \label{fig:sensitivity_analysis}
\end{figure}


\subsection{Comparative study}

In this subsection, we compare optimal capacity expansion decisions across different scenarios. We examine two airports, one with an initial capacity of 1 runway ($K_0=40$) and another with 2 runways ($K_0=80$). Each is evaluated under varying demand growth rates: slow ($\eta = 0.01$), moderate ($\eta = 0.015$), and fast ($\eta = 0.02$). Additionally, we assess the impact of different levels of uncertainty: low uncertainty ($\sigma = 0.02, \lambda = 0.02, Z_i = -10\%$), medium uncertainty($\sigma = 0.04, \lambda = 0.04, Z_i = -15\%$), and high uncertainty ($\sigma = 0.06, \lambda = 0.06, Z_i = -20\%$). Table \ref{tab:comp_study} summarizes the optimal decisions for each scenario.

\begin{table}[ht]
\caption{Optimal decisions for the comparative study.}
\centering
{\begin{tabular}{cccccccc}
\toprule
\multirow{3}{*}{$K_0$}                                                     & \multirow{3}{*}{$\eta$} & \multicolumn{6}{c}{Uncertainty level}                                                          \\ \cmidrule{3-8} 
                                                                          &                      & \multicolumn{2}{c}{Low}         & \multicolumn{2}{c}{Med}         & \multicolumn{2}{c}{High} \\ \cmidrule{3-8} 
                                                                          &                      & $q^*$ & \multicolumn{1}{c}{$\Delta K^*$} & $q^*$ & \multicolumn{1}{c}{$\Delta K^*$} & $q^*$        & $\Delta K^*$       \\ \midrule
\multirow{3}{*}{\begin{tabular}[c]{@{}c@{}}40\\ (1 runway)\end{tabular}}  & 0.01                 & 11.3  & \multicolumn{1}{c}{40}       &  12.0 & \multicolumn{1}{c}{40}        &     12.9     &      40         \\
                                                                          & 0.015                 & 11.2  & \multicolumn{1}{c}{40}        &   11.8& \multicolumn{1}{c}{40}        &     12.7     &     40          \\
                                                                          & 0.02                 & 13.0  & \multicolumn{1}{c}{80}        &  13.5 & \multicolumn{1}{c}{80}        &      14.4    &      80         \\ \midrule
\multirow{3}{*}{\begin{tabular}[c]{@{}c@{}}80\\ (2 runways)\end{tabular}} & 0.01                 &  20.1 & \multicolumn{1}{c}{40}        & 21.4  & \multicolumn{1}{c}{40}        &    22.9      &    40           \\
                                                                          & 0.015                 & 22.4  & \multicolumn{1}{c}{80}        & 23.6  & \multicolumn{1}{c}{80}        &      25.3    &        80       \\
                                                                          & 0.02                 & 25.9  & \multicolumn{1}{c}{160}        & 25.4 & \multicolumn{1}{c}{120}        &   27.0       &      120         \\ \bottomrule
\end{tabular}}
\label{tab:comp_study}
\end{table}

From the table, we observe that as uncertainty increases from left to right, the trigger demand $q^*$ increases if the expansion size $\Delta K$ keeps unchanged; while $q^*$ decreases if $\Delta K$ decreases as well (for example, in the last row, as the uncertainty level increases from low to medium, the expansion size decreases from 160 operations per hour to 120 operations per hour, while the trigger demand decreases from 25.9 operations per hour to 25.4 operations per hour). This reflects the requirement for a higher demand level to prove the profitability of a capacity expansion decision under greater uncertainty. Meanwhile, a higher uncertainty level generally calls for a more conservative approach, resulting in a smaller expansion size.  Additionally, when moving from top to bottom for a given initial capacity $K_0$, the relationship between $q^*$ and $\Delta K$ is different. An increased demand growth rate $\eta$ generally leads to a larger expansion size and an increase in trigger demand; however, if the expansion size remains unchanged, the trigger demand tends to decrease slightly (for example, given $K_0 = 40$ and low uncertainty level, as the demand growth rate $\eta$ increases from 0.01 to 0.015, the expansion size remains unchanged at 40 operations per hour, while the trigger demand decreases from 11.3 operations per hour to 11.2 operations per hour). Finally, all else being equal, a larger initial capacity $K_0$ always implies a higher trigger demand and a larger expansion size. These findings offer valuable insights for airport authorities in planning runway capacity expansion decisions.

\section{Conclusions and future studies}
\label{sec: Conclusions and Future Studies}
In conclusion, this study aims to guide investment decisions for airport runway capacity expansion. Specifically, we answer the questions of when and by how much to expand runway capacity for public airports under demand uncertainty. To avoid the delay cost being unrealistic large when demand approaches or even exceeds capacity, a modified delay function is proposed. Two models are developed: a deterministic dynamic model and a stochastic dynamic model. In the stochastic model, we incorporate two levels of uncertainty: first, random fluctuations; and second, random jumps, caused by crisis events. This research fills a significant gap in the existing literature by quantitatively evaluating how different sources of uncertainty impact airport capacity expansion decisions, an area that is often overlooked. Additionally, our findings provide valuable insights for optimized management strategies. The main findings are summarized as follows:
\begin{enumerate}
    \item As the initial capacity increases, both the trigger demand and the expansion size increase. A similar trend is observed when the volatility rate increases. At points where the expansion size experiences a jump, the trigger demand also shows a spike. While this pattern holds with an increased demand growth rate when both the trigger demand and the expansion size increase, if the expansion size remains unchanged, the trigger demand decreases slightly.
    \item Random jumps have a distinct impact on capacity expansion decisions. A larger jump size or higher jump probability increases the trigger demand if the expansion size remains unchanged. However, as the jump size or jump probability continues to increase, the expansion size eventually drops, leading to a corresponding drop in the trigger demand.
    \item Runway capacity expansion decisions vary across airports with different initial capacities, demand growth rates, and uncertainty levels. Generally, a larger initial capacity implies both a higher trigger demand and a larger expansion size. A higher level of uncertainty always encourages a more conservative capacity expansion decision.
\end{enumerate}

Compared with \cite{balliauw2020expanding}, which investigated the optimal capacity expansion decision for private airports, our study, while yielding similar results in most cases, uncovers some new insights. First, two layers of uncertainty play distinct roles in the capacity expansion decision. The expansion size increases as the volatility rate increases, but decreases when the random jump size or probability increases. Regarding the trigger demand, it consistently increases with the volatility rate, but decreases when the expansion size drops as the random jump size or probability increases. Second, with the same expansion size, a higher demand growth rate leads to a lower trigger demand. However, when the expansion size jumps to the next level as the demand growth rate increases, an upward jump in the trigger demand is also observed.

For future studies, several potential directions can be explored. First, the runway system may be influenced by various factors, such as its physical layout and weather conditions. A more detailed investigation may be required to better reflect the cost structure of the airport runway system. Besides, alternative formulations of the delay function and system cost function may be considered to provide a more accurate representation of the airport runway cost structure. Second, while this study focuses solely on the runway system, it may be beneficial to consider other airport components, such as aprons, gates, and terminals, in conjunction. Additionally, assessing the profits generated by an airport alongside its capacity expansion decisions could provide insights aimed at maximizing profitability for private operators, further enhancing the results provided in \cite{balliauw2020expanding}. Third, capacity expansion decisions heavily rely on future demand projections. Although we adopt the widely used assumption that demand follows an exponential trend, as modeled by GBM, it is important to acknowledge that a carrying limit may exist for the maximum potential demand level. Therefore, the analyses presented in this paper apply primarily to small- or medium-sized airports where demand grows exponentially. This assumption may not hold when demand has already reached large levels. If that happens, utilizing the stochastic logistic process for demand growth modeling may offer a more accurate modeling approach, as suggested by \cite{li2024trb}. Other stochastic processes may also be considered for demand modeling. Fourth, this study only considers a one-time expansion decision; however, in many cases, such decisions are made sequentially. Future work could investigate the interrelationship of each expansion decision within the context of sequential decision-making, as illustrated by \cite{bensoussan2019sequential}. Lastly, this study assumes that jump occurrences follow a deterministic dynamic process. However, future jumps could be correlated, with their magnitudes influenced by potential demand. Incorporating such correlations into the analysis would represent a valuable methodological contribution.

\section*{Acknowledgments}

This work was presented at the 6th Bridging Transportation Researchers (BTR 6) online conference (\cite{li2024btr}). The authors thank the attendees for their constructive comments and suggestions.

\section*{Disclosure of interest}
No potential conflict of interest was reported by the author(s).

\bibliographystyle{elsarticle-harv}
\bibliography{ref}

\appendix

\section{Calculations of Eqs. (\ref{eq:deterministic_expected_cumulative_cost_difference}) and (\ref{eq:expected_cumulative_cost_difference})}
\subsection{Calculations of Eq. (\ref{eq:deterministic_expected_cumulative_cost_difference})}
\label{sec: cal_detail_1}
\renewcommand{\thetheorem}{A\arabic{theorem}} 
Calculation shows that
\begin{equation}
\resizebox{0.89\textwidth}{!}{%
$
\begin{aligned}
    F(q, \Delta K)
    &= TC_1(T) - TC_2(T, \Delta K)\\
    &= \int_{T}^{\infty}N_p\left(C_1(\widetilde{Q}_t) - C_2(\widetilde{Q}_t, \Delta K)\right) e^{-\rho (T - t)} dt  - (f + v \Delta K)\\
    &= \int_{0}^{\infty}N_p\left(C_1(\widetilde{Q}_{t+T}) - C_2(\widetilde{Q}_{t+T}, \Delta K)\right) e^{-\rho t} dt  - (f + v \Delta K)\\    
    &= \int_{0}^{\infty}\left[N_p\left(C_1(Q_t) - C_2(Q_t, \Delta K)\right)  - \rho(f + v \Delta K)\right]  e^{-\rho t} dt\\
    &= \int_{0}^{\infty}\left[\Delta C(Q_t, \Delta K)  e^{-\rho t} dt\right],
\end{aligned}
$%
}
\end{equation}
where $Q_t = \widetilde{Q}_{t+T}$. Here, $\widetilde{Q}$ is the demand growth process starting from an initial demand different from the demand at the investment time $T$, with $\widetilde{Q}_{T} = q$; while $Q$ is the process with an initial demand that equals the demand at the investment timing ($Q_0 = q$). The cost difference function $\Delta C(q, \Delta K)$ is defined in Eq. (\ref{eq:deterministic_cost_difference_rate}).

\subsection{Calculations of Eq. (\ref{eq:expected_cumulative_cost_difference})}
\label{sec: cal_detail_2}

\begin{lemma}
\label{lemma:strong_markov}
(The strong Markov property for Ito diffusion, \cite{oksendal2013stochastic}, Theorem 7.2.4) For a bounded Borel function $f$, if $\tau$ is a stopping time, $Q$ is an Ito process, then
\begin{equation}
    \mathbb{E}_{|Q(\tau)=q}\left[f(Q(t + \tau))\right] = \mathbb{E}_{|Q(0)=q}\left[f(Q(t))\right].
\end{equation}
\end{lemma}

Eq. (\ref{eq:expected_cumulative_cost_difference}) can be calculated as follows:
\begin{equation}
\resizebox{0.89\textwidth}{!}{%
$
\begin{aligned}
    F (q, \Delta K)
    &= \mathbb{E}_{|Q(\tau)=q}\left[TC_1(\tau) - TC_2(\tau, \Delta K)\right]\\
    &= \mathbb{E}_{|Q(\tau)=q}\left[\int_{\tau}^{\infty}N_p\left(C_1(Q(t)) - C_2(Q(t), \Delta K)\right) e^{-\rho (t-\tau)} dt  - (f + v \Delta K) \right]\\
    &= \mathbb{E}_{|Q(\tau)=q}\left[\int_{0}^{\infty}N_p\left(C_1(Q(t + \tau)) - C_2(Q(t + \tau), \Delta K)\right)  e^{-\rho t} dt   - (f + v \Delta K)\right]\\
    &= \mathbb{E}_{|Q(0)=q}\left[\int_{0}^{\infty}N_p\left(C_1(Q(t)) - C_2(Q(t), \Delta K)\right)  e^{-\rho t} dt   - (f + v \Delta K)\right]\\
    &= \mathbb{E}_{|Q(0)=q}\int_{0}^{\infty}\left[N_p\left(C_1(Q(t)) - C_2(Q(t), \Delta K)\right) - \rho(f + v \Delta K)\right]  e^{-\rho t} dt\\
    &= \mathbb{E}_{|Q(0)=q}\int_{0}^{\infty}\left[\Delta C(Q(t), \Delta K)  e^{-\rho t} dt\right],
\end{aligned}
$%
}
\end{equation}
where the cost difference function $\Delta C(q, \Delta K)$ is defined in Eq. (\ref{eq:deterministic_cost_difference_rate}). The fourth equation in Eq. (\ref{eq:expected_cumulative_cost_difference}) has applied the strong Markov property from Lemma \ref{lemma:strong_markov}.

\section{Proof of Theorem \ref{thm: F_deterministic}}
\label{sec: proof_F_deterministic}

Given $Q_0 = q$, we define $t_1$ and $t_2$ such that $Q_{t_1}=qe^{\eta t_1} = K_0, Q_{t_2}=qe^{\eta t_2} = K_0 + \Delta K$. By solving these equations, we find $t_1 = \frac{1}{\eta}\ln{\frac{K_0}{q}}, t_2 = \frac{1}{\eta}\ln{\frac{K_0 + \Delta K}{q}}$. Substituting Eq. (\ref{eq:deterministic_cost_difference_rate}) into Eq. (\ref{eq:deterministic_expected_cumulative_cost_difference}) results in the cumulative cost difference function $F(q, \Delta K)$ as follows:
\begin{equation}
\label{eq:deterministic_F_function}
\resizebox{0.89\textwidth}{!}{%
$
\begin{aligned}
    ~&~F(q, \Delta K)\\
    =&~ \int_{0}^{t_1} \left( \alpha_1 Q_t^{\beta + 1}+ \alpha_2 \right) e^{-\rho t} dt +
    \int_{t_1}^{t_2} \left(\alpha_3 Q_t^{\beta + 1} + \alpha_4 Q_t + \alpha_5 \right) e^{-\rho t} dt +
    \int_{t_2}^{\infty} \alpha_6 e^{-\rho t} dt\\
    =&~ \int_{0}^{t_1} \left( \alpha_1 q^{\beta + 1}e^{\eta(\beta + 1)t}+ \alpha_2 \right) e^{-\rho t} dt +
    \int_{t_1}^{t_2} \left(\alpha_3 q^{\beta + 1}e^{\eta(\beta + 1)t} + \alpha_4 qe^{\eta t}  + \alpha_5 \right) e^{-\rho t} dt +
    \int_{t_2}^{\infty} \alpha_6 e^{-\rho t} dt\\
    =&~ \int_{0}^{t_1} \left( \alpha_1 q^{\beta + 1}e^{(\eta(\beta + 1) - \rho)t}+ \alpha_2e^{-\rho t} \right) dt + \int_{t_1}^{t_2} \left(\alpha_3 q^{\beta + 1}e^{(\eta(\beta + 1) - \rho)t} + \alpha_4 qe^{(\eta-\rho) t} + \alpha_5 e^{-\rho t}\right)  dt + \int_{t_2}^{\infty} \alpha_6 e^{-\rho t} dt\\
    =&~ \left(\frac{\alpha_1 q^{\beta + 1}(e^{(\eta(\beta + 1) - \rho)t_1} - 1)}{\eta(\beta + 1) - \rho} - \frac{\alpha_2 (e^{-\rho t_1} - 1)}{\rho}\right) \\
    &\qquad +
    \left(\frac{\alpha_3 q^{\beta + 1}(e^{(\eta(\beta + 1) - \rho)t_2} - e^{(\eta(\beta + 1) - \rho)t_1})}{\eta(\beta + 1) - \rho} + \frac{\alpha_4 q(e^{(\eta - \rho)t_2} - e^{(\eta - \rho)t_1})}{\eta- \rho} - \frac{\alpha_5 (e^{-\rho t_2} - e^{-\rho t_1})}{\rho}\right) \\
    &\qquad +
    \frac{\alpha_6 e^{-\rho t_2}}{\rho}\\
    =&~ \left(\frac{\alpha_1 q^{\beta + 1}\left(\left(\frac{K_0}{q}\right)^{\frac{\eta(\beta + 1) - \rho}{\eta}} - 1\right)}{\eta(\beta + 1) - \rho} - \frac{\alpha_2 \left(\left(\frac{K_0}{q}\right)^{-\frac{\rho}{\eta}} - 1\right)}{\rho}\right) \\
    &\qquad +
    \left(\frac{\alpha_3 q^{\beta + 1}\left(\left(\frac{K_0 +\Delta K}{q}\right)^{\frac{\eta(\beta + 1) - \rho}{\eta}} - \left(\frac{K_0}{q}\right)^{\frac{\eta(\beta + 1) - \rho}{\eta}}\right)}{\eta(\beta + 1) - \rho} + \frac{\alpha_4 q\left(\left(\frac{K_0 +\Delta K}{q}\right)^{\frac{\eta - \rho}{\eta}} - \left(\frac{K_0}{q}\right)^{\frac{\eta - \rho}{\eta}}\right)}{\eta- \rho} - \frac{\alpha_5 \left(\left(\frac{K_0 +\Delta K}{q}\right)^{-\frac{\rho}{\eta}} - \left(\frac{K_0}{q}\right)^{-\frac{\rho}{\eta}}\right)}{\rho}\right)  \\
    &\qquad +
    \frac{\alpha_6 \left(\frac{K_0 + \Delta K}{q}\right)^{-\frac{\rho}{\eta}}}{\rho}\\
    =&~ \left(\frac{\alpha_3  \left(\frac{K_0 +\Delta K}{q}\right)^{(\beta+1)-\frac{\rho}{\eta}} + (\alpha_1-\alpha_3) \left(\frac{K_0 }{q}\right)^{(\beta+1)-\frac{\rho}{\eta}} - \alpha_1}{\eta(\beta + 1) - \rho}\right) q^{\beta + 1}  \\
    &\qquad + \frac{\alpha_4 \left(\left(\frac{K_0 +\Delta K}{q}\right)^{1-\frac{\rho}{\eta}} - \left(\frac{K_0}{q}\right)^{1-\frac{\rho}{\eta}}\right)}{\eta- \rho}q\\
    &\qquad + \frac{ (\alpha_6 - \alpha_5)\left(\frac{K_0 +\Delta K}{q}\right)^{-\frac{\rho}{\eta}}    +  (\alpha_5 - \alpha_2)\left(\frac{K_0}{q}\right)^{-\frac{\rho}{\eta}} + \alpha_2 }{\rho}\\
    =&~ \left( \frac{\alpha_3\left(K_0+\Delta K\right)^{(\beta+1)-\frac{\rho}{\eta}} + (\alpha_1-\alpha_3) K_0^{(\beta+1)-\frac{\rho}{\eta}} }{\eta(\beta+1)-\rho}   \right) q^{\frac{\rho}{\eta}} -\frac{\alpha_1}{\eta(\beta+1)-\rho} q^{\beta+1}  \\
    &\qquad + \frac{\alpha_4 \left(\left(K_0 +\Delta K\right)^{1-\frac{\rho}{\eta}} - K_0^{1-\frac{\rho}{\eta}}\right)}{\eta- \rho}q^{\frac{\rho}{\eta}}\\
    &\qquad + \frac{ (\alpha_6 - \alpha_5)\left(K_0 +\Delta K\right)^{-\frac{\rho}{\eta}}    +  (\alpha_5 - \alpha_2) K_0^{-\frac{\rho}{\eta}}}{\rho} q^{\frac{\rho}{\eta}} + \frac{\alpha_2}{\rho}\\
    =&~ A_1q^{\beta + 1} + A_2q^{\frac{\rho}{\eta}}+A_3,
\end{aligned}$%
}
\end{equation}
where $A_1, A_2, A_3$ are given as follows:
\begin{equation}
    \resizebox{0.89\textwidth}{!}{%
    $
    \begin{aligned}
        A_1 &= - \frac{\alpha_1}{\eta(\beta+1)-\rho} = - \frac{N_pA\alpha}{\eta(\beta+1)-\rho}\left({K_0^{-\beta}}-\left(K_0+\Delta K\right)^{-\beta}\right),\\
        A_2 &= \frac{\alpha_3  \left(K_0 +\Delta K\right)^{(\beta+1)-\frac{\rho}{\eta}} + (\alpha_1-\alpha_3) K_0^{(\beta+1)-\frac{\rho}{\eta}}}{\eta(\beta + 1) - \rho} + \frac{\alpha_4 \left(\left(K_0 +\Delta K\right)^{1-\frac{\rho}{\eta}} - K_0^{1-\frac{\rho}{\eta}}\right)}{\eta- \rho} + \frac{ (\alpha_6 - \alpha_5)\left(K_0 +\Delta K\right)^{-\frac{\rho}{\eta}}    +  (\alpha_5 - \alpha_2)K_0^{-\frac{\rho}{\eta}}}{\rho}\\
        &= - \frac{N_pA\alpha\left(\left(K_0+\Delta K\right)^{1-\frac{\rho}{\eta}} - K_0^{1-\frac{\rho}{\eta}}\right)}{\eta(\beta+1)-\rho} 
        +\frac{N_p A \alpha(\beta + 1)\left( \left(K_0+\Delta K\right)^{1-\frac{\rho}{\eta}} - K_0^{1-\frac{\rho}{\eta}} \right)}{\eta-\rho}
        + \frac{N_p A \alpha\beta\left(\left(K_0+\Delta K\right)^{1-\frac{\rho}{\eta}} - K_0^{1-\frac{\rho}{\eta}}\right)}{\rho}\\
        &= N_pA\alpha\left(\frac{\beta + 1}{\eta-\rho} + \frac{\beta}{\rho} -\frac{1}{\eta(\beta+1)-\rho} \right)\left(\left(K_0+\Delta K\right)^{1-\frac{\rho}{\eta}} - K_0^{1-\frac{\rho}{\eta}}\right),\\
        A_3 &= \frac{\alpha_2}{\rho} = -\frac{(N_pc_h+\rho v)\Delta K + \rho f}{\rho}.
    \end{aligned}$%
    }
\end{equation}

\section{Proof of Theorem \ref{thm:F_function}}
\label{sec: proof_F_stochastic}

\setcounter{theorem}{0} 
\renewcommand{\thetheorem}{C\arabic{theorem}} 

\begin{lemma}
\label{lemma:ito}
(Ito's lemma, \cite{dixit1994investment}, page 86, Eq. (42)) For the general stochastic process
\begin{equation}
\label{eq:general_stochastic_process_with_jump}
    dQ = a(Q,t) dt + b(Q,t) dw + c(Q,t) d\sum_{i=1}^{N_t} Z_i,
\end{equation}
with $w, Z_i, N_t$ defined the same as Eq. (\ref{eq:sde_basic}), if $F$ is a differentiable function of $Q$ and $t$,
\begin{equation}
\label{eq:ito_with_jump_general}
\resizebox{0.89\textwidth}{!}{%
    $
    \mathbb{E}[dF] = \left[\frac{\partial F}{\partial t} + a(Q,t)\frac{\partial F}{\partial Q} + \frac{1}{2} b^2(Q,t) \frac{\partial^2 F}{\partial Q^2}\right]dt + \mathbb{E}_Z\{\lambda[F(Q + c(Q,t)Z, t) - F(Q,t)]\}.
    $%
}
\end{equation}
\end{lemma}

Since $F(Q(t), \Delta K)$ satisfies the following dynamic programming formula:
\begin{equation}
\label{eq:cost_difference_dp}
\resizebox{0.89\textwidth}{!}{%
    $
    F(Q(t), \Delta K) = \Delta C(Q(t), \Delta K) dt + e^{-\rho dt}\mathbb{E}[F(Q(t), \Delta K) + d F(Q(t), \Delta K)],
$%
}
\end{equation}
applying Lemma \ref{lemma:ito} for $F(Q(t))$ with $Q(t)$ defined by Eq. (\ref{eq:sde_basic}), when setting $\Delta K$ fixed, gives
\begin{equation}
\label{eq:ito_sde}
\resizebox{0.89\textwidth}{!}{%
    $
    \mathbb{E}[d F(Q(t))] = \left[\eta Q(t)\frac{d F(Q(t))}{d Q} + \frac{1}{2} \sigma^2Q(t)^2 \frac{d^2 F(Q(t))}{d Q^2}\right]dt + \mathbb{E}_Z\left\{\lambda \left[F(Q(t)(1+Z)) - F(Q(t))\right]\right\} dt,
    $%
    }
\end{equation}
where for simplicity we omit $\Delta K$ in function $F(Q(t), \Delta K)$. If $Q(t) = q$, substituting Eq. (\ref{eq:ito_sde}) into Eq. (\ref{eq:cost_difference_dp}) yields the following ordinary differential equation:
\begin{equation}
\label{eq:F_ode_by_dp_and_ito}
\resizebox{0.89\textwidth}{!}{%
    $
    \frac{1}{2} \sigma^2q^2 F''(q) + \eta q F'(q) + \lambda \mathbb{E}_Z \left[F(q(1+Z)) - F(q)\right] -\rho F(q) + \Delta C(q, \Delta K) = 0,
    $%
}
\end{equation}
which has to be satisfied by the solution for values of $q$ inside three regions: region $R_1$: $q < K_0$; region $R_2$: $K_0 \leq q < K_0 + \Delta K$; and region $R_3$: $q \geq K_0 + \Delta K$.

The solution of Eq. (\ref{eq:F_ode_by_dp_and_ito}) is as follows (\cite{dangl1999investment}):
\begin{equation}
\label{eq:F_solution_by_cauchy_euler}
    F(q)_{q \in R_i} = F_i(q) = A_{i, 1} q^{b_1} + A_{i, 2} q^{b_2} + \overline{F_i}(q),\quad i = 1,2,3.
\end{equation}
$\overline{F_i}$ is a particular solution of the differential equation (\ref{eq:F_ode_by_dp_and_ito}) with $\Delta C = \Delta C_i$ (defined in Eq. (\ref{eq:deterministic_cost_difference_rate})). The solution of $\overline{F_i}$ is given by Eq. (\ref{eq:ode_particular_sol}).

The condition for $b$ is:
\begin{equation}
\label{eq:solving_b}
    \varphi(b) := \frac{\sigma^2}{2}b(b-1) + \eta b + \lambda \mathbb{E}_Z (1+Z)^{b} -(\rho+\lambda) = 0.
\end{equation}
Note that by definition, $Z$ is the jump size (percentage), which is a random variable taking values less than zero but greater than $-1$. Consequently, we conclude that $0 \leq 1+ Z \leq 1$ holds almost surely. Therefore, $\mathbb{E}_Z (1+Z)^{b}$ is bounded, leading us to observe $\lim_{b\rightarrow-\infty}\varphi(b) = +\infty$ and $\lim_{b\rightarrow+\infty}\varphi(b) = +\infty$. Considering that $\varphi(b) \leq \frac{\sigma^2}{2}b(b-1) + \eta b -\rho$, we find $\varphi(0) \leq -\rho < 0$ and $\varphi(1) \leq \eta -\rho < 0$ when the discount rate $\rho$ is larger than the demand growth rate $\eta$. Therefore, Eq. (\ref{eq:solving_b}) has two solutions, $b_1$ and $b_2$, with $b_1 > 1$ and $b_2 < 0$.
Following the approach in \cite{dangl1999investment}, since demand $Q(t)$ is volatile, it is able to cross the boundaries freely. Consequently, the solution must satisfy the following boundary conditions:
\begin{align}
    \label{eq:value-matching1}F(K_0-) &= F(K_0+),\\
    \label{eq:value-matching2}F((K_0 + \Delta K)-) &= F((K_0 + \Delta K)+),\\
    \label{eq:smooth-pasting1}F_q'(K_0-) &= F_q'(K_0+),\\
    \label{eq:smooth-pasting2}F_q'((K_0 + \Delta K)-) &= F_q'((K_0 + \Delta K)+).
\end{align}

Additionally, when considering large values of $q$, particularly when $q\in R_3$ and grows large, the likelihood of $q$ crossing the boundary into $R_2$ diminishes. Therefore,
\begin{equation}
    \lim_{q\rightarrow\infty} F(q) = \mathbb{E}\left[\int_0^\infty \alpha_6 e^{-\rho t}\right] = \frac{\alpha_6}{\rho} = \overline{F_3}.
\end{equation}
Since $b_1 > 1, b_2 < 0$, we conclude that
\begin{equation}
    \label{eq:boundary_infty}
    A_{3,1} = 0.
\end{equation}

Finally, we consider the characteristics of the stochastic process modeled as GBM combined with a jump process. For GBM, 0 is an absorbing point (i.e., if $Q(0)=0$, then $Q(t) \equiv 0$ for any time $t$). When incorporating jumps, 0 remains an absorbing point because a jump represents a proportionate decrease in demand. Therefore, 
\begin{equation}
    F(0) = \mathbb{E}\left[\int_0^\infty(0 + \alpha_2)e^{-\rho t}\right] = \frac{\alpha_2}{\rho}.
\end{equation}
Given that $b_1 > 1, b_2 < 0$, we conclude that
\begin{equation}
    \label{eq:boundary_0}
    A_{1,2} = 0.
\end{equation}
Solving Eqs. (\ref{eq:value-matching1}), (\ref{eq:value-matching2}), (\ref{eq:smooth-pasting1}), (\ref{eq:smooth-pasting2}), (\ref{eq:boundary_infty}), (\ref{eq:boundary_0}) yields the values of $A_{i,j}, i\in\{1,2,3\}, j\in\{1,2\}$, as shown in Eq. (\ref{eq:A_ij}).
    \begin{equation}
    \label{eq:A_ij}
    \resizebox{0.88\textwidth}{!}{%
    $
    \begin{aligned}
        A_{1,1} &= A_{2,1} + A_{2,2} K_0^{b_2-b_1} + \frac{\overline{F_2}(K_0) - \overline{F_1}(K_0)}{K_0^{b_1}},
        \\
        A_{1,2} &= 0,\\
        A_{2,1} &= \frac{b_2(\overline{F_3}(K_0 + \Delta K) - \overline{F_2}(K_0 + \Delta K)) - (K_0 + \Delta K)(\overline{F_3}'(K_0 + \Delta K) - \overline{F_2}'(K_0 + \Delta K))}{(b_2-b_1)(K_0 + \Delta K)^{b_1}},\\
        A_{2,2} &= \frac{b_1(\overline{F_2}(K_0) - \overline{F_1}(K_0)) - K_0(\overline{F_2}'(K_0) - \overline{F_1}'(K_0))}{(b_2-b_1)K_0^{b_2}},\\
        A_{3,1} &= 0,\\
        A_{3,2} &= A_{2,1}(K_0+\Delta K)^{b_1-b_2} + A_{2,2} - \frac{\overline{F_3}(K_0 + \Delta K) - \overline{F_2}(K_0 + \Delta K)}{(K_0 + \Delta K)^{b_2}}.
    \end{aligned}
    $%
    }
    \end{equation}
    
    We therefore obtain the function $F(q, \Delta K)$.

\section{Proof of Theorem \ref{thm: solving_V} and Theorem \ref{thm:trigger}}
\label{sec: proof_trigger}

Analogous to the proof of Theorem \ref{thm:F_function}, we apply Ito's lemma to $V(Q(t))$:
\begin{equation}
\label{eq:ito_V}
\resizebox{0.89\textwidth}{!}{%
    $
    \mathbb{E}[d V(Q(t))] = \left[\eta Q(t)\frac{d V(Q(t))}{d Q} + \frac{1}{2} \sigma^2Q(t)^2 \frac{d^2 V(Q(t))}{d Q^2}\right]dt + \mathbb{E}_Z\left\{\lambda \left[V(Q(t)(1+Z)) - V(Q(t))\right]\right\} dt.
    $%
    }
\end{equation}
Substituting Eq. (\ref{eq:ito_V}) into Eq. (\ref{eq:option_pred}) results in the following ordinary differential equation for $V$:
\begin{equation}
\label{eq:ode_V}
    \frac{1}{2} \sigma^2q^2 V''(q) + \eta q V'(q) + \lambda \mathbb{E}_Z \left[V(q(1+Z)) - V(q)\right] -\rho V(q) = 0.
\end{equation}
The solution of Eq. (\ref{eq:ode_V}) is as follows:
\begin{equation}
\label{eq:V_solution_by_cauchy_euler}
    V(q) = \overline{A_1} q^{b_1} + \overline{A_2} q^{b_2}.
\end{equation}
Since $V$ cannot approach infinity when $q$ is at the absorbing point 0, the coefficient of $q^{b_2}$, denoted as $\overline{A_2}$, must be zero. Therefore, we have:
\begin{equation}
\label{eq:v_boundary_a2}
    \overline{A_2} = 0.
\end{equation}
To solve for $\overline{A_1}$, we consider the no arbitrage condition: the transition from the waiting area to the investment area when $Q(t) = q^*$ must be value-matched. If we define 
\begin{equation}
    \overline{F}(q) = \max_{\Delta K} F(q, \Delta K),
\end{equation}
then the value-matching condition is given by
\begin{equation}
\label{eq:v_boundary_a1_raw}
    \overline{F}(q^*) = V(q^*).
\end{equation}
The conditions in Eqs. (\ref{eq:v_boundary_a2}) and (\ref{eq:v_boundary_a1_raw}) provide the solution to Eq. (\ref{eq:V_solution_by_cauchy_euler}), if the optimal solution $q^*$ is known. To determine $q^*$, 
we should recognize that at the transition point, the option value must be maximized. Therefore, $q^*$ satisfies the following condition:
\begin{equation}
\label{eq:q*_condition_raw}
    q^*= \argmax_{q}\left\{\frac{\overline{F}(q)}{q^{b_1}}\right\}.
\end{equation}
Consequently, we have:
\begin{equation}
\label{eq:A1_condition}
    \overline{A_1} = \max_{q}\left\{\frac{\overline{F}(q)}{q^{b_1}}\right\}.
\end{equation}

It can be proved that maximizing the option value is equivalent to the smooth-pasting condition below (see Appendix \ref{sec: value_matching}). 
\begin{equation}
\label{eq:value_matching_condition}
\overline{F}'(q^*) = V'(q^*).
\end{equation}

Solving Eqs. (\ref{eq:v_boundary_a2}), (\ref{eq:v_boundary_a1_raw}) and (\ref{eq:A1_condition}) yields the option function in Eq. (\ref{eq:V_solution_by_cauchy_euler}) along with the trigger value $q^*$. Subsequently, the optimal expansion size is given by Eq. (\ref{eq:expansion_size_stoc}).

\section{Equivalence of Eqs. (\ref{eq:A1_condition}) and (\ref{eq:value_matching_condition})}
\label{sec: value_matching}

Let us define
\begin{equation}
 \phi(q) = \frac{\overline{F}(q)}{q^{b_1}}.
\end{equation}
To maximize $\phi(q)$, we apply the first order condition:
\begin{equation}
\label{eq:foc}
    q^* \overline{F}'(q^*) - b_1 \overline{F}(q^*) = 0.
\end{equation}
Considering
\begin{equation}
    V(q) = A_1 q^{b_1}, \quad V'(q) = A_1b_1q^{b_1-1},
\end{equation}
we can express
\begin{equation}
\label{eq:relation_V_V_prime}
    b_1 V(q^*) = q V'(q^*).
\end{equation}
Combining Eq. (\ref{eq:relation_V_V_prime}) with
\begin{equation}
    \overline{F}(q^*) = V(q^*),
\end{equation}
we find
\begin{equation}
    F'(q^*) = \frac{b_1 \overline{F}(q^*)}{q^*} = \frac{b_1V(q^*)}{q^*} = V'(q^*).
\end{equation}

The derivations above are reversible, thus proving the equivalence of Eqs. (\ref{eq:A1_condition}) and (\ref{eq:value_matching_condition}).

\end{document}